\documentclass[11pt]{elsarticle}
\usepackage{stmaryrd}
\usepackage{amsmath,amssymb,amsthm}
\usepackage[colorlinks,linkcolor=red,anchorcolor=blue,citecolor=green]{hyperref}
\usepackage{tikz}
\usepackage{subfigure}
\usepackage{graphicx}

\theoremstyle{plain}
\newtheorem{theorem}{Theorem}[section]
\newtheorem{lemma}[theorem]{Lemma}

\newtheorem{proposition}[theorem]{Proposition}
\newtheorem{corollary}[theorem]{Corollary}

\theoremstyle{definition}
\newtheorem{definition}[theorem]{Definition}
\newtheorem{remark}[theorem]{Remark}
\newtheorem{example}[theorem]{Example}

\numberwithin{equation}{section}

\newcommand{\dda}{\mathord{\mbox{\makebox[0pt][l]{\raisebox{-.4ex}{$\downarrow$}}$\downarrow$}}}

\newcommand{\da}{\mathord{\downarrow}}
\newcommand{\rom}[1]{\rm{\uppercase\expandafter{\romannumeral #1}}}

\newcommand{\set}[2]{\{#1\mid#2\}}

\makeatletter
\def\ps@pprintTitle{%
  \let\@oddhead\@empty
  \let\@evenhead\@empty
  \def\@oddfoot{\reset@font\hfil\thepage\hfil}
  \let\@evenfoot\@oddfoot
}
\makeatother

\begin{document}

\begin{frontmatter}

\title{Continuous Domains in Formal Concept Analysis\tnoteref{t1}}
\tnotetext[t1]{Supported by the National Natural Science Foundation of China(11771134).}

\author[]{Longchun Wang$^{a,b}$}
\address[1]{College of Mathematics and Econometrics, Hunan University, Changsha, Hunan, 410082, China}
\address[2]{School of Mathematical Sciences, Qufu Normal University, Qufu, Shandong, 273165, China}

\author[3]{Lankun Guo\corref{a1}}
\address[3]{College of Mathematics and Computer Sciences, Hunan Normal University, Changsha, Hunan, 410012, China}
\author[1]{Qingguo Li\corref{a1}}
\cortext[a1]{Corresponding author.}
\ead{lankun.guo@gmail.com(Lankun Guo);liqingguoli@aliyun.com}

\begin{abstract}

Formal Concept Analysis has proven to be an effective method of restructuring complete lattices and various algebraic domains. In this paper, the notions of attribute continuous formal context and continuous formal concept are introduced by considering a selection~$\mathcal{F}$ of finite subsets of attributes. Our  decision of a selection~$\mathcal{F}$ relies on a kind of generalized interior operators.  It is shown that the set of  continuous formal concepts forms a continuous domain, and every continuous domain can be obtained in this way. Moreover, an  notion of $\mathcal{F}$-morphism is also identified to produce a category  equivalent to that of continuous domains with Scott-continuous functions. This paper also consider the representations of various subclasses of continuous domains such as algebraic domains, bounded complete domains and stably continuous semilattices. These results explore the fundamental idea of domain theory in Formal Concept Analysis from a categorical viewpoint.

\end{abstract}

\end{frontmatter}
\section{Introduction}

The notion of formal concept evolved in the philosophical
theory and logical theory.  In the early 1980s, to better understand lattice theory for potential users of lattice-based methods for data management, Formal Concept Analysis (FCA) was initiated by Wille\cite{ganter99}.
 A central theme of
FCA is  to restructure  lattice theory by formal concept.
\emph{The basic theorem on concept lattices}~\cite{ganter99} has shown that the set of formal concepts ordered by set inclusion forms a complete lattice, called a concept lattice, and that every complete lattice can be restructured as a concept lattice. Since then, FCA  has developed into a interdisciplinary research area  with a thriving theoretical community and an increasing number of
applications in computer science and artificial intelligence~\cite{guo11, lai09, li17, li17a, poelmans13a}.

Domain theory was introduced by Scott as an abstract mathematical model of formal languages. The key idea of domain theory is partiality and approximation, which  makes sure that a higher order object can be successively approximated by computing those ordinary objects. Then domain theory can be specified as a computationally relevant framework  and as such has found applications in the theory of denotational semantics as well as in aspects of knowledge representation and reasoning.  An important topic in domain theory  is to build the interrelation between  domains and various mathematical structures. As can be expected, a great deal of  effort has gone into characterizations of various domains.
  To name a few examples, we have Scott's representations of Scott domains as information systems~\cite{Scott82}, Abramsky's flexible work  \emph{domain theory in logical form} for SFP-domains \cite{abramsky91} and its extensions \cite{Chen06, jung99},  some topological investigations by Vickers~\cite{Vickers89}.  More articles about this issue can be found in \cite{jung13, Lawson97,  Spreen12, Spreen08}.

In~\cite{zhang06}, Zhang and Shen brought the two independent areas, FCA and Domain theory, to a common meeting point. They generalized the notion of concept to approximable concept and obtained a representation theorem of algebraic lattices based on approximable concepts. Hitzler et al.~\cite{hitzler06} built  the notion of approximable concept with a category which is equivalent to that of algebraic lattices and Scott-continuous functions.
In \cite{huang14}, Huang et al. presented a notion of $F$-approximble concept and provided  an approach to representing algebraic domains.  Almost at the same time, Guo et al. \cite{guo14} developed two variations of rough approximable concepts and obtained the representation of algebraic domains by FCA in the rough setting. These works enrich the link between domain theory and FCA. For further information on the relationship between domain theory and FCA, see, for example, \cite{guo18, hitzler,  lei09, li13a}.

However, all these generalizations of  formal concepts mentioned above  represent only subclasses of algebraic domains. And  general continuous  domains are highly relevant to lots of tasks in data analysis and knowledge reasoning, covering important example based on the real line. To the best of our knowledge, the intimate relationship between general continuous domains and FCA were not explicated until the current paper. The aim of this paper is to explore the continuity in FCA from a category-theoretical viewpoint and give a novel approach to representing continuous domains by means of FCA. For this purpose, we  generalize the technique of classical formal concepts and propose the notions of attribute continuous formal context and continuous formal attribute concept. A new tool used in our definitions is a family of nonempty finite subsets of attributes in a formal context. It is shown that attribute continuous formal concepts generate exactly continuous domains. In order to represent the Scott-continuous functions between continuous domains, a notion of $\mathcal{F}$-morphism between attribute continuous formal contexts is introduced. Then the category of attribute continuous formal contexts with $\mathcal{F}$-morphisms is  established which is equivalent to that  of continuous domains with Scott-continuous functions. Furthermore, the representations of many subclasses of continuous domains are studied by driving some appropriate conditions into an attribute continuous formal context. Especially, we capture stably continuous semilattices, a special kinds of continuous semilattices, in term of FCA.
  All these results demonstrate the capacity of FCA in representing continuous domains.

The paper is organized as follows: In Section 2, the necessary definitions and results from domain theory and FCA are recalled. In Section~3, notions of attribute continuous formal  contexts and continuous formal concepts are introduced, which generalized formal contexts and approximable concepts in some sense. It is proved that each continuous attribute formal  contexts generate a continuous domain. For a great variety of subclasses of continuous domains, how they can be represented as attribute continuous formal  contexts is also studied.
In Section~4,  the appropriate morphism for continuous attribute formal  contexts is investigated. Then the categorical equivalence between continuous domains and attribute continuous formal  contexts is established. Some remarks can be found in Section 5.
\section{Preliminaries}
\subsection{Domain theory}
Let us first recall some basic definitions and results of domain theory. Our main references will be~\cite{ davey02, gierz03, goubault13a}.

  A \emph{poset} $P$ is a set equipped with a reflexivity, antisymmetry and transitivity binary relation $\leq$ on it. If a poset~$(P,\leq)$ has a least element, it is called \emph{pointed}. A semilattice is a poset in which every two elements $x,y$ have a greatest lower bound $x\wedge y$. For any subset $A$ of $P$, the \emph{down set}~$\da A$ of $A$  is a set $\set{x\in P}{\exists a\in A, x\leq a}$. We  abbreviate a \emph{principal idea}~$\da \{x\}$ as $\da x$. A subset  $D$ of  $P$ is \emph{directed} if it is nonempty and every finite subset of $D$ has an upper bound in $D$. A poset $P$  is said to be a \emph{dcpo} if  each directed subset~$D$ of  $P$ has a  least upper bound~$\bigvee D$ in  $P$.  A  \emph{complete lattice} is a poset $P$ in which every subset has a least upper bound.

Let $P$ be a dcpo. Then $x$ is \emph{ way below} $y$, in notation $x\ll y$, if and only if for any directed subset $D$ of $P$, the relation $y\leq \bigvee D$ always implies that $x\leq d$ for some $d\in D$. Obviously, $x\ll y$ implies that $x\leq y$.
For any $X\subseteq P$, $\dda X$ is defined to be the set
$\set{y\in P}{ (\exists x\in X)y\ll x}$. And $\dda \{x\}$  is abbreviated as $\dda x$.
   An element $x\in P$ is said to be \emph{compact} if $x\ll x$. We write $K(P)$ to stand for the set of compact elements of $P$. A \emph{basis} $B_P$ of $P$ is a subset of $P$ such that,
for every $x\in P$, the collection $B_P \cap \dda  x$  is directed, and has $x$ as a least upper bound.
  \begin{definition}\label{d2.1}
   \begin{enumerate}[(1)]
  \item A dcpo $P$ is said to be a \emph{continuous domain} if  it has a basis.
  \item A \emph{bounded complete domain} is a continuous domain in which every bounded above subset has a least upper bound.
  \item A \emph{continuous lattice} is a continuous domain which is also a complete lattice.
  \item A \emph{stably continuous semilattice} is a continuous domain such that it is a semilattice and the way below relation on it is \emph{multiplicative}, that is, $x\ll y,z$ implies $x\ll \wedge z$.
   \item A  dcpo $P$ is said to be an \emph{algebraic domain} if $K(P)$ forms a basis.
   \end{enumerate}
   \end{definition}

The way-below relation on a continuous domain  satisfies the following \emph{interpolation property}:
\begin{proposition}
Let $P$ be a continuous domain and let $M \subseteq P$ be a finite set with $M\ll y$. Then there exists $z\in P$ such that $M \ll z \ll y$ holds, where $M\ll y$ means that $x\ll y$ for all elements $x\in M$.
\end{proposition}
\begin{definition}\label{d21}
A function~$f:P\rightarrow Q$ between two continuous domains is said to be \emph{Scott-continuous} if, for all directed subset $D\subseteq P$, $f(\bigvee D)=\bigvee_{x\in D}f(x)$.
\end{definition}
\subsection{Formal concept analysis}

 A fundamental contribution of FCA is to provide  an alternative formulation of lattice theory, which starts from the notions of formal context and formal concept. We highly recommend \cite{ganter99} as an introduction to classical FCA.

A \emph{formal context}  is a triple~$(P_o,P_a,\vDash_P)$ where  $P_o$ is a set of \emph{objects} and $P_a$ is a set of \emph{attributes}. The  relation $\vDash_P$ is a subset of $P_o\times P_a$.
In this case, two functions can be
 defined:
 \begin{equation}\label{e2.1}
 \alpha:\mathcal{P}(P_o)\rightarrow \mathcal{P}(P_a), A\mapsto\set{n\in P_a}{\forall m\in A, m\vDash n},
 \end{equation}
\begin{equation}\label{e2.2}
\omega:\mathcal{P}(P_a)\rightarrow \mathcal{P}(P_o), B\mapsto\set{m\in P_o}{\forall n\in B, m\vDash n}.
\end{equation}
A \emph{formal (attribute) concept} of a formal context $(P_o,P_a,\vDash)$ is a subset $B\subseteq P_a$ which is a fixed-point of $\alpha\circ\omega$. Dually, a \emph{formal object concept} of a formal context $(P_o,P_a,\vDash)$ is a subset $A\subseteq P_o$ if it  is a fixed point of $\omega\circ\alpha$.
The set of all formal attribute concepts and the set of all formal object concepts of a formal context $(P_o,P_a,\vDash)$ are denoted by $\mathfrak{B}(P_a)$ and $\mathfrak{B}(P_o)$, respectively. 

For any set $X$, let $\mathcal{P}(X)$ and $\mathcal{F}(X)$ denote the powerset of  $X$ and the family of all finite subsets of $X$, respectively. The notion $F\sqsubseteq X$  means that $F\in\mathcal{F}(X)$.

A \emph{closure operator} on $X$ is a function $\gamma$ on $\mathcal{P}(X)$ which is  extensive ($A\subseteq \gamma(A)$), monotone ($A\subseteq B\Rightarrow \gamma(A)\subseteq \gamma(B)$) and idempotent ($\gamma(A)=\gamma(\gamma(A))$).
An \emph{interior operator} on $X$ is a function $\tau$ on $\mathcal{P}(X)$ which is contractive ($\tau(A)\subseteq A$), monotone and idempotent.

\begin{proposition}[The basic theorem on concept lattices]~\cite{ganter99}
Let $(P_o,P_a,\vDash)$ be a formal context.
 \begin{enumerate}[(1)]
 \item Both $(\mathfrak{B}(P_a),\subseteq)$ and $(\mathfrak{B}(P_o),\subseteq)$ are complete lattices, and are anti-isomorphic to each other.
 \item The compositions $\alpha\circ\omega$ and $\omega\circ\alpha$ are closure operators.
 \end{enumerate}
\end{proposition}

Much research has shown that we can also restructure some order structures by means of FCA. For example, Zhang and Shen \cite{zhang06} represented algebraic lattices by approximable  concepts, Huang et al. \cite{huang14} developed the notion of $F$-approximable  concept to generate algebraic domains.
\begin{definition}~\label{d2.5}\cite{zhang06}
Let $(P_o,P_a,\vDash)$ be a formal context. A subset $Q\subseteq P_a$ is called an \emph{approximable  concept} if the following condition holds,
\begin{enumerate}[{\bf (AC)}]
\item $M\sqsubseteq Q\Rightarrow\alpha(\omega(M))\subseteq Q.$
\end{enumerate}
\end{definition}

\begin{definition} \cite{huang14}\label{d2.6}
Let  $(P_o,P_a,\vDash)$ be a formal context and $\mathcal{F}$  a nonempty finite subset of $P_a$ which satisfies the following condition,
 \begin{enumerate}[{(\bf FC)}]
 \item $(\forall F\in \mathcal{F}) M\sqsubseteq \alpha\circ \omega(F)\Rightarrow (\exists G\in \mathcal{F})M\subseteq G\subseteq \alpha\circ \omega(F).$
 \end{enumerate}
Then $(P_o,P_a,\vDash, \mathcal{F})$ is called a \emph{conditional formal context}. And a subset $Q\subseteq P_a$ is called an \emph{F-approximable concept} if it satisfies the following conditions,
\begin{enumerate}[{(\bf F{A}1)}]
\item $(\forall M \sqsubseteq Q)(\exists F\in \mathcal{F}) M\subseteq F\subseteq Q,$
\item $(\forall F\in \mathcal{F}) F\subseteq Q \Rightarrow \alpha(\omega(F))\subseteq Q.$
\end{enumerate}
\end{definition}

Given a formal context~$(P_o,P_a,\vDash)$, let $\mathcal{F}=\mathcal{F}(P_a)$. Then by Definition~\ref{d2.6}, $(P_o,P_a,\vDash,\mathcal{F})$ is a conditional formal  context. And for any subset~$Q$ of $P_a$, it is an approximable concept of~$(P_o,P_a,\vDash)$ if and only if it is  an F-approximable  concept of~$(P_o,P_a,\vDash,\mathcal{F})$.

\section{Representations of various domains}
The main purposes of this section is to establish a systematic connection between continuous domains and FCA.
\subsection{Continuous formal concept}
In this subsection, based on a formal context~$(P_o,P_a,\vDash)$ and a selection $\mathcal{F}$ of finite subsets of $P_a$, we present the notion of continuous formal concept, and investigate the relationship between this notion and those of approximable  concept, F-approximable  concept and classical formal  concept.  A new technique is that the selection $\mathcal{F}$ relies on a  generalized interior operator introduced in the following.

\begin{definition}~\label{d5.5}
Let $\mathbb{P}=(P_o,P_a,\vDash)$ be a formal context. A \emph{kernel attribute operator} on $\mathbb{P}$ is a mapping  $\tau:\mathcal{P}(P_a)\rightarrow\mathcal{P}(P_a)$ which fulfils the following conditions, for any  $B, B_1\subseteq P_a$,
\begin{enumerate}[{(\bf{A}1)}]
\item $\tau(\alpha(\omega(B)))\subseteq \alpha(\omega(B))$;
\item $\tau(\tau(\alpha(\omega(B))))=\tau(\alpha(\omega(B)))$;
\item $\tau(\alpha(\omega(B)))\subseteq\tau(\alpha(\omega(B_1)))$ whenever $B\subseteq B_1$.
\end{enumerate}
\end{definition}

It is easy to see that each interior operator on $P_a$, particularly the identity map~\rm{id}$_{\mathcal{P}(P_a)}$,  is a  kernel attribute operator. 

 Given a kernel attribute operator $\tau$ on a formal context~$(P_o,P_a,\vDash)$, we often abbreviate $\tau\circ\alpha\circ\omega(B)$ as $\lceil B\rceil$ for any $B\subseteq P_a$. Condition~(A3) indicates that the operator~$\lceil\cdot\rceil$ is monotone. Moreover, it is idempotent.  In fact, since $\lceil\lceil B\rceil\rceil=\tau\circ\alpha\circ\omega\circ\tau\circ\alpha\circ\omega(B)$ and $\alpha\circ\omega$ is a closure operator on $P_a$, we have
 $$\lceil\lceil B\rceil\rceil\subseteq\tau\circ\alpha\circ\omega\circ\alpha\circ\omega(B)=
 \tau\circ\alpha\circ\omega(B)=\lceil B\rceil=\tau\circ\tau\circ\alpha\circ\omega(B)\subseteq\lceil\lceil B\rceil\rceil.$$ And hence for any $B\subseteq P_a$,
 \begin{equation}\label{e3.1}
 B_1\subseteq \lceil B\rceil\Rightarrow\lceil B_1\rceil\subseteq\lceil B\rceil.
 \end{equation}

 \begin{definition}\label{d3.2}
  Let $(P_o,P_a,\vDash)$  be a formal context and $\tau$ a kernel attribute operator on $(P_o,P_a,\vDash)$. A $\tau$-\emph{consistent selection} is a nonempty family $\mathcal{F}$  of nonempty subsets of $P_a$ which satisfies the following condition,
\begin{enumerate}[{\bf (C{A}1)}]
\item  $(\forall F\in \mathcal{F}) M\sqsubseteq\lceil F\rceil\Rightarrow (\exists G\in \mathcal{F}) (M\subseteq\lceil G\rceil, G\subseteq \lceil F\rceil)$.
\end{enumerate}
\end{definition}

In what follows, we denote a $\tau$-consistent selection $\mathcal{F}$
by $\mathcal{F}_{\tau}$.  For any $F\in\mathcal{F}_{\tau}$, since $F\neq \emptyset$, by condition~(CA1), it is clear that $\lceil F\rceil\neq \emptyset$ and $\set{\lceil G\rceil}{G\in \mathcal{F}_{\tau}, G\subseteq\lceil F\rceil}\neq \emptyset$.

\begin{definition} A structure $\mathbb{P}_{\tau}=(P_o,P_a,\vDash,\mathcal{F}_{\tau})$ is said to be an \emph{attribute continuous formal context} if $(P_o,P_a,\vDash)$  is a formal context and $\mathcal{F}_{\tau}$  is a $\tau$-consistent selection.
\end{definition}
 For any family $\mathcal{F}$  of nonempty subsets of $P_a$, perhaps there exist two different kernel attribute operators~$\tau$ and $\upsilon$ on $(P_o,P_a,\vDash)$ such that $\mathcal{F}$ is both a $\tau$-consistent selection  and a $\upsilon$-consistent selection. In this case, $\mathbb{P}_{\tau}=(P_o,P_a,\vDash,\mathcal{F}_{\tau})$ and $\mathbb{P}_{\upsilon}=(P_o,P_a,\vDash,\mathcal{F}_{\upsilon})$ are two different  attribute continuous formal contexts.
\begin{proposition}\label{p3.4}
In Definition~\ref{d3.2}, condition (CA1) is equivalent to saying that for any $F\in \mathcal{F}_{\tau}$, there exists a directed set $\set{\lceil F_i\rceil}{F_i\in\mathcal{F}_{\tau},F_i\subseteq \lceil F\rceil,i\in I}$ such that $\lceil F\rceil=\bigcup_{i\in I}\lceil F_i\rceil.$
\end{proposition}
\begin{proof}
Suppose that condition (CA1) holds. For any $F\in\mathcal{F}_{\tau}$, as has been mentioned,  $\set{\lceil G\rceil}{G\in\mathcal{F}_{\tau}, G\subseteq \lceil F\rceil}$ is not empty.
Let $G_1,G_2\in \mathcal{F}_{\tau}$ and $G_1,G_2\subseteq\lceil F\rceil$. Then $G_1\cup G_2\sqsubseteq \lceil F\rceil$. By condition~(CA1), there exists some $G_3\in\mathcal{F}_{\tau}$ such that $G_3\subseteq\lceil F\rceil$ and $G_1\cup G_2\subseteq \lceil G_3 \rceil$. This implies that $\lceil G_1\rceil\subseteq \lceil G_3 \rceil$ and $\lceil G_2 \rceil\subseteq \lceil G_3 \rceil$. This show that $\set{\lceil G\rceil}{G\in\mathcal{F}_{\tau}, G\subseteq \lceil F\rceil}$ is directed. We now show that $\lceil F\rceil=\bigcup\set{\lceil G\rceil}{G\in\mathcal{F}_{\tau}, G\subseteq \lceil F\rceil}$.
Since  $\bigcup\set{\lceil G\rceil}{G\in\mathcal{F}_{\tau}, G\subseteq \lceil F\rceil}\subseteq \lceil F\rceil$ is clear, it suffices to prove the reverse inclusion. For any $x\in \lceil F\rceil$, by condition~(CA1), there exists some $G_x\in\mathcal{F}_{\tau}$ such that $G_x\subseteq\lceil F\rceil$ and $\{x\}\subseteq \lceil G_x \rceil$. From $G_x\subseteq\lceil F\rceil$, it follows that $\lceil G_x\rceil\subseteq\lceil F\rceil$, and thus $\lceil F\rceil=\bigcup\set{\{x\}}{x\in\lceil F\rceil}\subseteq \bigcup\set{\lceil G_x\rceil}{x\in\lceil F\rceil}\subseteq \bigcup\set{\lceil G\rceil}{G\subseteq\lceil F\rceil}$. 

Conversely, suppose $F\in \mathcal{F}_{\tau}$ and there exists a directed $\set{\lceil F_i\rceil}{F_i\in\mathcal{F}_{\tau},F_i\subseteq \lceil F\rceil,i\in I}$ such that $\lceil F\rceil=\bigcup_{i\in I}\lceil F_i\rceil$. Then for any $M\sqsubseteq\lceil F\rceil=\bigcup_{i\in I}\lceil F_i\rceil$, we have some $j\in I$ such that $F_j\in\mathcal{F}_{\tau}$, $F_j\subseteq \lceil F\rceil$ and $M\subseteq \lceil F_j\rceil$.  Hence condition~(CA1) follows.
\end{proof}
\begin{definition} Let $\mathbb{P}_{\tau}=(P_o,P_a,\vDash,\mathcal{F}_{\tau})$ be an attribute continuous formal context. A \emph{continuous formal (attribute) concept} of $\mathbb{P}_{\tau}$ is a subset~$Q$ of $P_a$ which satisfies the following condition,
\begin{enumerate} [{\bf (CA2)}]
\item  $ M\sqsubseteq Q\Rightarrow (\exists F\in \mathcal{F}_{\tau}) M\subseteq\lceil F\rceil\subseteq Q$.
\end{enumerate}
\end{definition}
We denote the set of all the continuous formal concepts of
by $\mathfrak{B}(P_a,\mathcal{F}_{\tau})$. Taking $M=\emptyset$, condition~(CA2) yields that every continuous formal concept $Q$ of $\mathbb{P}_{\tau}$ is not empty.

 The following two propositions tell us that the notion of continuous formal concept is a generalization of approximable  concept as well as of F-approximable concept.

\begin{proposition}

Each conditional formal context~$(P_o,P_a,\vDash,\mathcal{F})$ can induce an attribute continuous formal context~$(P_o,P_a,\vDash,\mathcal{F}_{\tau})$, where $\mathcal{F}_{\tau}=\mathcal{F}$. Moreover, a subset~$Q\subseteq P_a$ is an F-approximable concept of $(P_o,P_a,\vDash,\mathcal{F})$ if and only if it is  a continuous formal concept of $(P_o,P_a,\vDash,\mathcal{F}_{\tau})$.
\end{proposition}
\begin{proof} Let $(P_o,P_a,\vDash,\mathcal{F})$ be a conditional formal context, and  $\tau$
 the identity map on ${\mathcal{P}(P_a)}$. Then $\tau$ is a kernel attribute operator on $(P_o,P_a,\vDash)$ and  $\lceil B\rceil= \tau\circ\alpha\circ\omega(B)=\alpha\circ\omega(B)$ for any $B\subseteq P_a$. Thus by condition~(FC) and the fact that $\alpha\circ \omega$ is a closure operator on $P_a$, condition~(CA1) follows. Therefore, $\mathcal{F}$ is a $\tau$-consistent selection, and hence $(P_o,P_a,\vDash,\mathcal{F}_{\tau})$ is an attribute continuous formal context.

 If $Q$ is an $\mathcal{F}$-approximable concept of~$(P_o,P_a,\vDash,\mathcal{F})$ and $M\sqsubseteq Q$. Then by condition~(FA1), there exists some $F\in \mathcal{F}$ such that $M\subseteq F\subseteq Q$. It follows from condition~(FA2) that $\alpha(\omega(F))=\lceil F\rceil\subseteq Q$, which implies that  condition~(CA2) holds. Therefore, $Q$ is a continuous formal  concept of attribute continuous formal context~$(P_o,P_a,\vDash,\mathcal{F}_{\tau})$.

 Conversely, suppose that $Q$ is a continuous formal  concept of ~$(P_o,P_a,\vDash,\mathcal{F}_{\tau})$. To finish the proof, it suffices to prove that $Q$ satisfies conditions~(FA1) and (FA2). For condition~(FA1), let $M\sqsubseteq Q$. By condition~(CA2), $M\subseteq\lceil G\rceil\subseteq Q$ for some $G\in \mathcal{F}_{\tau}$. Since $Q$ is a continuous formal  concept of ~$(P_o,P_a,\vDash,\mathcal{F}_{\tau})$, there exists some $F\in \mathcal{F}_{\tau}$ such that $M\subseteq F \subseteq \lceil G\rceil$. Consequently, $M\subseteq F\subseteq Q$. For condition~(FA2), suppose that $F\in\mathcal{F}_{\tau}$ and $F\subseteq Q$. Then by condition~(CA2), we have some $G\in \mathcal{F}_a$ with $F\subseteq \lceil G\rceil= \alpha\circ\omega(G)\subseteq Q$. Since $\alpha\circ\omega$ is a closure operator, it follows that $\alpha\circ\omega(F)\subseteq Q$.
\end{proof}

\begin{proposition}\label{p3.7}
Each  formal context~$(P_o,P_a,\vDash)$ can induce an attribute continuous formal context~$(P_o,P_a,\vDash,\mathcal{F}_{\tau})$, where $\mathcal{F}=\mathcal{F}(P_a)$. Moreover, a subset~$Q\subseteq P_a$ is an approximable  concept of $(P_o,P_a,\vDash)$ if and only if it is  a continuous formal  concept of $(P_o,P_a,\vDash,\mathcal{F}_{\tau})$.
\end{proposition}
\begin{proof}
Suppose that $(P_o,P_a,\vDash)$ is a formal context. Let $\tau$ be
 the identity map on ${\mathcal{P}(P_a)}$ and $\mathcal{F}=\mathcal{F}(P_a)$. It is easy to see that $\tau$ is a kernel attribute operator on $(P_o,P_a,\vDash)$ and  $\mathcal{F}$ is a $\tau$-consistent selection. Then $(P_o,P_a,\vDash,\mathcal{F}_{\tau})$ is an attribute continuous formal context, which is  said to be  \emph{the induced attribute continuous formal context by $(P_o,P_a,\vDash)$}.

 Let $Q$  be an approximable  concept of $(P_o,P_a,\vDash)$. For any $M\sqsubseteq Q$, since $\alpha\circ\omega$ is a closure operator on $P_a$ and $\tau$ is the identity map on $\mathcal{P}(P_a)$, we have
 $M\subseteq \alpha(\omega(M))=\tau(\alpha(\omega(M)))=\lceil M\rceil\subseteq Q.$
 Note that $M\in \mathcal{F}_{\tau}$, it follows that condition~(CA2) holds, and then  $Q$ is  a continuous formal  concept of $(P_o,P_a,\vDash,\mathcal{F}_{\tau})$. Conversely, suppose that  $Q$ is  a continuous formal concept of $(P_o,P_a,\vDash,\mathcal{F}_{\tau})$. Then for any $M\sqsubseteq Q$, there exists some $F\in\mathcal{F}_{\tau}$ such that $M\subseteq \lceil F\rceil=\alpha(\omega(F))\subseteq Q$, which implies that $\alpha(\omega(M))\subseteq Q$. Therefore, condition~(AC) follows and $Q$  is an approximable  concept of $(P_o,P_a,\vDash)$.
\end{proof}

\begin{remark}\label{r3.8}
Consider a formal context~$(P_o,P_a,\vDash)$ in which  $P_a$ is finite. Then formal concepts and approximable  concepts coincide. So a subset $Q$ of  $P_a$ is a formal  concept of $(P_o,P_a,\vDash)$ if and only if it is  a continuous formal  concept of $(P_o,P_a,\vDash,\mathcal{F}_{\tau})$, where $(P_o,P_a,\vDash,\mathcal{F}_{\tau})$ is the induced attribute continuous formal context by $(P_o,P_a,\vDash)$.
\end{remark}

 The finiteness of $P_a$ is necessary for Remark~\ref{r3.8}.

\begin{example}
Consider the complete lattice~$(L_1,\leq)$ in  Figure~1, where the least upper bound of the set $\{a_1,a_2,a_3,\cdots\}$ is $\top$. Then~$(L_1,L_1,\geq_1)$ is a formal context. 
 Let $(L_1,L_1,\geq_1,\mathcal{F}_{\tau_1})$ be the induced attribute continuous formal context by $(L_1,L_1,\geq)$. It is not difficult to see that the set~$\{a_1,a_2,a_3,\cdots\}\cup\{\bot\}$ is a continuous formal concept of $(L_1,L_1,\geq_1,\mathcal{F}_{\tau_1})$ but not a formal concept of $(L_1,L_1,\geq_1)$. And the concept lattice $(\mathfrak{B}(L_1),\subseteq)$
 is isomorphic to $(L_1,\leq)$ but not isomorphic to $(\mathfrak{B}(L_1,\mathcal{F}_{\tau_1}),\subseteq)$.  For  formal context~$(L_2,L_2,\geq_2)$, the set~$\{a_1',a_2',a_3',\cdots\}\cup\{\bot'\}$ is a continuous formal concept of $(L_2,L_2,\geq_2,\mathcal{F}_{\tau_2})$  but not a formal concept of $(L,L,\geq)$.

In addition, as a consequence of Theorem~\ref{t3.14},  there is no
 attribute continuous formal context~$(P_o,P_a,\vDash,\mathcal{F}_{\tau})$ such that
 $(\mathfrak{B}(P_a,\mathcal{F}_{\tau}),\subseteq)$ is isomorphic to $(L_1,\leq_1)$,
 but there is an attribute continuous formal context~$(P_o,P_a,\vDash,\mathcal{F}_{\tau})$ such that
 $(\mathfrak{B}(P_a,\mathcal{F}_{\tau}),\subseteq)$ is isomorphic to $(L_2,\leq_2)$,
 as well as the concept lattice $(\mathfrak{B}(L_2),\subseteq)$. 
 \begin{figure}[htbp]
\centering
\begin{tikzpicture}[scale=1.15]
 \path (-3.2,2) node[left] {$\top$} coordinate (-3.2,2);
\fill (-3.2,2) circle (1pt);
   \path (-3.2,0.8) node[left] {$a_{3}$} coordinate (a_{3});
\fill (a_{3})circle (1pt);
\path (-3.2,0.2) node[left] {$a_{2}$} coordinate (a_{2});
 \fill (a_{2})circle (1pt);
  \path (-3.2,-0.4) node[left] {$a_{1}$} coordinate (a_{1});
\fill (a_{1}) circle (1pt);
    \path (-1.5,1.5) node[right] {$b$} coordinate (b);
\fill (b) circle (1pt);
    \path (-3.2,-1) node[below] {$\bot$} coordinate (-3.2,-1);
\fill (-3.2,-1) circle (1pt);
\path (-2,-1) node[below] {$(L_1,\leq_1)$} coordinate (-2,-1);
 \draw (a_{1}) -- (-3.2,-1);
  \draw (a_{2}) -- (a_{1});
  \draw (a_{3}) -- (a_{2});
   \draw (b) -- (-3.2,-1);
    \draw (b) -- (-3.2,2);
  \draw [dotted][thick](a_{3}) -- (-3.2,2);

\path (3.2,2) node[left] {$\top{'}$} coordinate (-3.2,2);
\fill (3.2,2) circle (1pt);
\path (3.2,1.5) node[left] {$\top'_{1}$} coordinate (3.2,1.5);
\fill (3.2,1.5) circle (1pt);
   \path (3.2,0.8) node[left] {$a'_{3}$} coordinate (a_{3});
\fill (a_{3})circle (1pt);
\path (3.2,0.2) node[left] {$a'_{2}$} coordinate (a_{2});
 \fill (a_{2})circle (1pt);
  \path (3.2,-0.4) node[left] {$a'_{1}$} coordinate (a_{1});
\fill (a_{1}) circle (1pt);
    \path (4.9,1.5) node[right] {$b{'}$} coordinate (b);
\fill (b) circle (1pt);
    \path (3.2,-1) node[below] {$\bot'$} coordinate (3.2,-1);
\fill (3.2,-1) circle (1pt);
\path (4.5,-1) node[below] {$(L_2,\leq_2)$} coordinate (4.5,-1);
 \draw (a_{1}) -- (3.2,-1);
  \draw (a_{2}) -- (a_{1});
  \draw (a_{3}) -- (a_{2});
   \draw (b) -- (3.2,-1);
    \draw (b) -- (3.2,2);
    \draw (3.2,2) -- (3.2,1.5);
  \draw [dotted][thick](a_{3}) -- (3.2,1.5);
    \end{tikzpicture}
\caption{}
\label{fig 1}
\end{figure}
\end{example}
\subsection{Representation theorem of continuous domains}
In this subsection, we provide a new approach to representing continuous domains in term of FCA. Let us start by showing there are enough many  continuous formal concepts of~$\mathbb{P}_{\tau}$.
\begin{proposition}\label{p3.10}
Let $\mathbb{P}_{\tau}=(P_o,P_a,\vDash,\mathcal{F}_{\tau})$ be an attribute continuous formal context.  Then $\lceil F\rceil$ is a continuous formal concept of~$\mathbb{P}_{\tau}$ for any $F\in \mathcal{F}_{\tau}$.
\end{proposition}
\begin{proof}
 It is easy to see that $\lceil F\rceil$ satisfies condition~(CA2). 
\end{proof}

 The following proposition is frequently used later.
\begin{proposition}\label{p3.11}
Let $Q$ be a continuous formal concept of an attribute continuous formal context~$\mathbb{P}_{\tau}$.
\begin{enumerate}[(1)]
\item Then for any $M\sqsubseteq Q$, $\lceil M\rceil\subseteq Q$, and
\item there exists some $F\in \mathcal{F}_{\tau}$ such that $F\subseteq Q$ and $M\subseteq \lceil F\rceil\subseteq Q$.
\end{enumerate}
\end{proposition}
\begin{proof}
 Suppose that $Q$ is a continuous formal concept of~$\mathbb{P}_{\tau}$ with $M\sqsubseteq Q$. By condition~(CA2), $M\subseteq\lceil G\rceil\subseteq Q$ for some $G\in \mathcal{F}_{\tau}$. Then $\lceil M\rceil\subseteq\lceil G\rceil\subseteq Q$, and hence part~(1) follows.
  For part~(2), using condition~(AC1) to the above $G\in \mathcal{F}_{\tau}$ with $M\subseteq\lceil G\rceil$, there exists some $F\in \mathcal{F}_{\tau}$ such that $M\subseteq \lceil F\rceil$ and $F\subseteq\lceil G\rceil$. As consequence, $F\subseteq Q$ and $M\subseteq \lceil F\rceil\subseteq Q$.
\end{proof}

\begin{proposition}\label{p3.12}
Let $\mathbb{P}_{\tau}=(P_o,P_a,\vDash,\mathcal{F}_{\tau})$ be an attribute continuous formal context and $Q$ a subset of $P_a$. Then the following statements are equivalent:
\begin{enumerate} [(1)]
\item $Q$ is a continuous formal concept of~$\mathbb{P}_{\tau}$.
\item The set $\set{\lceil F\rceil}{F\in \mathcal{F}_{\tau}, F\subseteq Q}$ is directed and  $Q=\bigcup\set{\lceil F\rceil}{F\in \mathcal{F}_{\tau},  F\subseteq Q}$.
\end{enumerate}
\end{proposition}
\begin{proof}

(1) implies (2): Obviously, $\set{\lceil F\rceil}{F\in \mathcal{F}_{\tau}, F\subseteq Q}\neq \emptyset$. Let $F_1,F_2\in \mathcal{F}_{\tau}$ and $F_1,F_2\subseteq Q$. Then $F_1\cup F_2\sqsubseteq Q$. By part~(2) of Proposition~\ref{p3.11}, there exists some $F_3\in \mathcal{F}_{\tau}$ such that $F_3\subseteq Q$ and $F_1\cup F_2\subseteq \lceil F_3\rceil\subseteq Q$. Because $F_1\cup F_2\subseteq \lceil F_3\rceil$, it follows that $\lceil F_1\rceil\subseteq \lceil F_3\rceil$ and $\lceil F_2\rceil\subseteq \lceil F_3\rceil$.  This yields that the set $\set{\lceil F\rceil}{F\in \mathcal{F}_{\tau}, F\subseteq Q}$ is directed. According to part~(1) of Proposition~\ref{p3.11},  it is obvious that $\bigcup\set{\lceil F\rceil}{F\in \mathcal{F}_{\tau}, F\subseteq Q}\subseteq Q$. For the reverse inclusion, let $x\in Q$. Then by part~(2) of Proposition~\ref{p3.11}, we get some $F\in \mathcal{F}_{\tau}$ such that $F\subseteq Q$ and $\{x\}\subseteq\lceil F\rceil\subseteq Q$. Hence $Q\subseteq\bigcup\set{\lceil F\rceil}{F\in \mathcal{F}_{\tau}, F\subseteq Q}$.

(2) implies (1): Suppose that $M\sqsubseteq Q$.  Since the set $\set{\lceil F\rceil}{F\in \mathcal{F}_{\tau}, F\subseteq Q}$ is directed and $Q=\bigcup\set{\lceil F\rceil}{F\in \mathcal{F}_{\tau}, F\subseteq Q}$, there exists some $F\in \mathcal{F}_{\tau}$ such that $F\subseteq Q$ and $M\subseteq \lceil F\rceil\subseteq Q$. By condition~(CA1), we have $\lceil F\rceil\subseteq Q$. This implies that condition~(CA2) holds, and thus $Q$ is a continuous formal concept of~$\mathbb{P}_{\tau}$.
\end{proof}

Before stating the main result in this section, we
give a lemma which is of interest
in its own right.
\begin{lemma}\label{l3.13}
 Let $\mathbb{P}_{\tau}=(P_o,P_a,\vDash,\mathcal{F}_{\tau})$ be an attribute continuous formal context.
\begin{enumerate}[(1)]
\item For any directed subset~$\set{Q_i}{i\in I}$ of $\mathfrak{B}(P_a,\mathcal{F}_{\tau})$, the least upper bound $\bigvee_{i\in I}Q_i$ exists in $\mathfrak{B}(P_a,\mathcal{F}_{\tau})$, and $\bigvee_{i\in I}Q_i=\bigcup_{i\in I}Q_i$.
\item For all
 continuous formal concepts $Q_1$ and $Q_2$ of~$\mathbb{P}_{\tau}$, we have
\begin{equation}\label{e3.2}
Q_1\ll Q_2 \Leftrightarrow (\exists F\in\mathcal{F}_{\tau})(Q_1\subseteq \lceil F\rceil, F\subseteq Q_2).
\end{equation}
\item For any $F\in \mathcal{F}_{\tau}$,
\begin{equation}\label{e3.21}
\lceil F\rceil\ll \lceil F\rceil \Leftrightarrow (\exists G\in\mathcal{F}_{\tau})(\lceil F\rceil=\lceil G\rceil,G\subseteq \lceil G\rceil).
\end{equation}
\end{enumerate}
\end{lemma}
\begin{proof}
(1) It suffices to show that $\bigcup_{i\in I}Q_i$ is a continuous formal concept of $\mathbb{P}_{\tau}$.  Let $Q=\bigcup_{i\in I}Q_i$ and $M\sqsubseteq Q$. Then there exists some $j\in I$ such that $M\sqsubseteq Q_j$. Since $Q_j$ is a  continuous formal concept of $\mathbb{P}_{\tau}$, by condition~(CA2), there exists some $F\in \mathcal{F}_{\tau}$ with $M\subseteq \lceil F\rceil\subseteq Q_j\subseteq Q$. Thus, $Q$ is a continuous formal concept.

(2)
``$\Rightarrow $'', let $Q_1\ll Q_2 $. From  part~(2) of Proposition~\ref{p3.12},  we know that the set $\set{\lceil F\rceil}{F\in \mathcal{F}_{\tau}, F\subseteq Q_2}$ is directed and  $Q_2=\bigcup\set{\lceil F\rceil}{F\in \mathcal{F}_{\tau}, F\subseteq Q_2}$. Then by the definition of way-below relation, $Q_1\subseteq \lceil F\rceil$ for some $F\in \mathcal{F}_{\tau}$ with $F\subseteq Q_2$.

``$\Leftarrow $'', let $Q_2\subseteq \bigcup_{i\in I}U_i$, where $\set{U_i}{i\in I}$ is a directed subset of $\mathfrak{B}(P_a,\mathcal{F}_{\tau})$. Then there exists some $F\in \mathcal{F}_{\tau}$ such that $F\sqsubseteq Q_2\subseteq \bigcup_{i\in I}U_i$ and $Q_1\subseteq \lceil F\rceil$. Since$\set{U_i}{i\in I}$ is directed, there exists some $j\in I$ with $F\subseteq U_j$. Thus $Q_1\subseteq\lceil F\rceil\subseteq U_j$, which implies  $Q_1\ll Q_2 $.

(3) Let $F\in \mathcal{F}_{\tau}$ and
$\lceil F\rceil\ll \lceil F\rceil$. By Proposition~\ref{p3.10},  $\lceil F\rceil$ is a continuous formal concept.  It is clear  from part~(2) that part~(3) holds.
\end{proof}
\begin{remark}\label{r314}
Consider an element $F\in\mathcal{F}_{\tau}$. If $F\subseteq \lceil F\rceil$, then by Lemma~\ref{l3.13}, it is easy to see that $\lceil F\rceil\ll \lceil F\rceil$. Conversely, under the condition of $\lceil F\rceil\ll \lceil F\rceil$, we do not know whether $F$ is a subset of  $\lceil F\rceil$, but we have an element $G\in\mathcal{F}_{\tau}$ such that $\lceil F\rceil=\lceil G\rceil$ and $G\subseteq\lceil G\rceil$.
\end{remark}
\begin{theorem}[Representation theorem]\label{t3.14}
 Let $\mathbb{P}_{\tau}=(P_o,P_a,\vDash,\mathcal{F}_{\tau})$ be an attribute continuous formal context. Then $\mathfrak{B}(P_a,\mathcal{F}_{\tau})$ ordered by set inclusion forms a continuous domain.

 Conversely, for every continuous domain~$(D,\leq)$, there exists some attribute continuous formal context~$\mathbb{P}_{\tau}$ such that $(D,\leq)$ is order isomorphic to $(\mathfrak{B}(P_a,\mathcal{F}_{\tau}),\subseteq)$.
\end{theorem}
\begin{proof}

We first show that~$(\mathfrak{B}(P_a,\mathcal{F}_{\tau}),\subseteq)$ is a continuous domain. With part (1) of Lemma~\ref{l3.13}, $(\mathfrak{B}(P_a,\mathcal{F}_{\tau}),\subseteq)$ is clearly a dcpo. So that it
suffices to show that $(\mathfrak{B}(P_a,\mathcal{F}_a),\subseteq)$ has a basis. For any $Q\in \mathfrak{B}(P_a,\mathcal{F}_{\tau})$,  let $F\in \mathcal{F}_{\tau}$ with $F\subseteq Q$. Then by equation~(\ref{e3.2}), we have $\lceil F\rceil\ll Q$ in  $(\mathfrak{B}(P_a,\mathcal{F}_{\tau}),\subseteq)$. Proposition~\ref{p3.12} has proven that the set $\set{\lceil F\rceil}{F\in \mathcal{F}_{\tau}, F\subseteq Q}$ is directed and  $Q$ is its union. As a result, $\set{\lceil F\rceil}{F\in  \mathcal{F}_{\tau}}$ is a basis.

For the second part,
suppose that~$(D,\leq)$ is a continuous domain with a basis~$B_D$. Define the formal context $(P_o,P_a,\vDash)$, where $P_o=D$, $P_a=B_D$ and
$$x\vDash b\Leftrightarrow b\leq x.$$
Then by the definitions of $\omega$ and $\alpha$ introduced in section 2.2, for any $X\subseteq P_a$,  we have
\begin{align*}
\omega(X)&=\set{d\in D}{\forall x\in X, d\vDash x} \\
&=\set{d\in D}{\forall x\in X, x\leq d},
\end{align*}
and for any $Y\subseteq P_o$,
\begin{align*}
\alpha(Y)&=\set{b\in B_D}{\forall y\in Y, y\vDash b} \\
&=\set{b\in B_D}{\forall y\in Y, b\in \da y\cap B_D} \\
&=\bigcap\set{\da y\cap B_D}{ y\in Y}.
\end{align*}

For any $X\subseteq P_a$, define
\begin{equation}\label{e3.3}
\tau(X)=\dda X\cap B_D.
\end{equation}
Then
$\alpha(\omega(X))=\bigcap\set{\da y\cap B_D}{ y\in \set{d\in D}{\forall x\in X, x\leq d}}.$
Since $\dda X\subseteq \da X$, it is routine to check that the operator $\tau$ defined by equation~(\ref{e3.3}) is a kernel attribute operator on~$(P_o,P_a,\vDash)$. Let
\begin{equation}\label{e3.4}
\mathcal{F}=\set{F\sqsubseteq P_a}{\bigvee F\in F}.
\end{equation}
Then for any $F\in\mathcal{F}$, we have
\begin{equation}\label{e3.5}
\lceil F\rceil=B_D\cap\dda (\bigvee F).
\end{equation}

We now show that $\mathcal{F}$ is a $\tau$-consistent selection, and hence~$(P_o,P_a,\vDash,\mathcal{F}_{\tau})$ is an attribute continuous formal context. For any $F\in\mathcal{F}$, let $M\sqsubseteq \lceil F\rceil$. Since $\lceil F\rceil=B_D\cap\dda (\bigvee F)$, by the interpolation property of way-below relation, there exists some $b_0\in B_D$ such that $M\ll b_0\ll \bigvee F$. By setting $G=\{b_0\}$, condition~(CA1) follows.

We next prove that a subset~$Q$ of $P_a$ is a continuous formal attribute concept of $(P_o,P_a,\vDash,\mathcal{F}_{\tau})$ if and only if it  satisfies the following two conditions:
\begin{enumerate}[{\bf (R1)}]
\item $ (\forall u\in B_D )(\forall v\in Q)(u\ll v\Rightarrow u\in Q)$,
\item $(\forall M\sqsubseteq Q)(\exists u\in Q)M\ll u$.
\end{enumerate}

Suppose that $Q$ is a continuous formal attribute concept of $(P_o,P_a,\vDash,\mathcal{F}_{\tau})$. First, let $u\in B_D$ and $v\in Q$ with $u\ll v$. Then $\lceil \{v\}\rceil=\dda v\cap B_D\subseteq Q$. This implies that $u\in Q$. Second, if $M\sqsubseteq Q$, then there exists $F\in\mathcal{F}_D$ such that
$M\subseteq\lceil F\rceil\subseteq U$. Since $F\in \mathcal{F}_D$, it follows that
$\bigvee F\in F$. Define $u=\bigvee F$, then $M\ll u$.
Conversely, suppose that $Q\subseteq B_D$ which satisfies conditions~(R1) and (R2). For any $M\sqsubseteq Q$, there exists some $u\in Q$ such that $M\ll u$. Set $F=\{u\}$, then $F\in\mathcal{F}_D$ and $M\subseteq \dda u\cap B_D=\lceil F\rceil \subseteq Q$. Thus $Q$ is a continuous formal attribute concept of $(P_o,P_a,\vDash,\mathcal{F}_{\tau})$.

 For any $d\in D$, trivial checks verify that $\dda d \cap B_D$ satisfies conditions~(R1) and (R2), and thus it  is a continuous formal attribute concept of $(P_o,P_a,\vDash,\mathcal{F}_{\tau})$. On the other hand,  for any continuous formal attribute concept~$Q$ of $(P_o,P_a,\vDash,\mathcal{F}_{\tau})$, condition~(R2) indicates that $Q$ is a directed subset of $B_D$. Thus $\bigvee Q$ exists in $D$.  These allow us to define the following two functions:
$$f :D\rightarrow \mathfrak{B}(P_a,\mathcal{F}_{\tau}), x \mapsto \dda x\cap B_D,$$
$$g: \mathfrak{B}(P_a,\mathcal{F}_{\tau})\rightarrow D, Q \mapsto \bigvee Q.$$
Since it is not difficult to check that $f$ and $g$ are order preserving and mutually inverse, $(D,\leq)$ is isomorphic to $(\mathfrak{B}(P_a,\mathcal{F}_{\tau}),\subseteq)$, as required.
\end{proof}

Consider a continuous domain~$(D,\leq)$ with a basis $B_D$, the above theorem has associated it with an attribute continuous formal context~$(D,B_D,\geq,\mathcal{F}_{\tau})$, where the relation~$\geq$ is the dual order of $\leq$, further, $\tau$ and $\mathcal{F}_{\tau}$ are defined by equations~(\ref{e3.3}) and~(\ref{e3.4}), respectively.  In the sequel, we denote this associated attribute continuous formal context $(D,B_D,\geq,\mathcal{F}_{\tau})$ by $Rep(D)$.

\subsection{Representations of various subclasses of continuous domains}
In this subsection,  we investigate how to represent some important subclasses of continuous domains in term of FCA.

As well known, a large number of possible representations for  algebraic domains have been established. We add to this family  a characterization of algebraic domains based on the frame work of
attribute continuous formal contexts.
\begin{proposition}\label{p3.15}
Let $\mathbb{P}_{\tau}=(P_o,P_a,\vDash,\mathcal{F}_{\tau})$ be an attribute continuous formal context. Then
$(\mathfrak{B}(P_a,\mathcal{F}_{\tau}),\subseteq)$ is algebraic if and only if $\mathbb{P}_{\tau}$ satisfies the following additional condition,
\begin{enumerate}[{\bf(AD)}]
\item $(\forall F_1,F_2\in \mathcal{F}_{\tau})(F_2\subseteq \lceil F_1\rceil)\Rightarrow(\exists F\in\mathcal{F}_{\tau})(F_2\subseteq\lceil F\rceil, F\subseteq\lceil F\rceil, F\subseteq\lceil F_1\rceil)$.
\end{enumerate}
\end{proposition}
\begin{proof}
Theorem~\ref{t3.14} has shown that the set~$\set{\lceil F\rceil}{F\in \mathcal{F}_{\tau}}$ is a basis of $(\mathfrak{B}(P_a,\mathcal{F}_{\tau}),\subseteq)$. Thus the set  $K(\mathfrak{B}(P_a,\mathcal{F}_{\tau}))$ of all compact elements of $(\mathfrak{B}(P_a,\mathcal{F}_{\tau}),\subseteq)$ is a subset of $\set{\lceil F\rceil}{F\in \mathcal{F}_{\tau}}$. Further, for any $F\in\mathcal{F}_{\tau}$ with $F\subseteq\lceil F\rceil$, by equation~\ref{e3.21}, $\lceil F\rceil$ is a compact element of $\mathfrak{B}(P_a,\mathcal{F}_{\tau})$.

We first show the only if part,
suppose that $(\mathfrak{B}(P_a,\mathcal{F}_{\tau}),\subseteq)$ is an algebraic domain.  Let $F_1,F_2\in \mathcal{F}_{\tau}$ with $F_2\subseteq \lceil F_1\rceil$.
Then by equation~\ref{e3.2}, we have $\lceil F_2\rceil\ll \lceil F_1\rceil$. Since
$(\mathfrak{B}(P_a,\mathcal{F}_{\tau}),\subseteq)$ is algebraic,
 there exists some $G\in \mathcal{F}_{\tau}$ such that $\lceil F_2\rceil\ll \lceil G\rceil \ll \lceil G\rceil\ll \lceil F_1\rceil$. Then according to equation~\ref{e3.21}, $F_2\subseteq\lceil F\rceil$, $F\subseteq\lceil F\rceil$ and $ F\subseteq\lceil F_1\rceil$ for some $F\in \mathcal{F}_{\tau}$ with $\lceil F\rceil=\lceil G\rceil$.

Conversely, suppose that $\mathbb{P}_{\tau}$ satisfies condition~(AD). For any $Q\in \mathfrak{B}(P_a,\mathcal{F}_{\tau})$, put
$$K(Q)=\set{\lceil F\rceil}{F\in\mathcal{F}_{\tau},F\subseteq\lceil F\rceil,\lceil F\rceil\ll Q}.$$
It is not difficult to show that $K(Q)\subseteq K(\mathfrak{B}(P_a,\mathcal{F}_{\tau}))$ and
$$\set{\lceil F\rceil}{F\in\mathcal{F}_{\tau},F\subseteq\lceil F\rceil,\lceil F\rceil\ll Q}=\set{\lceil F\rceil}{F\in\mathcal{F}_{\tau},F\subseteq\lceil F\rceil\cap Q}.$$
 We now claim that the set $K(Q)$ is directed. In fact, let $\lceil G_1\rceil,\lceil G_2\rceil\in K(Q)$, where $G_1,G_2\in \mathcal{F}_{\tau}$ satisfy $G_1\subseteq\lceil G_1\rceil\cap Q$ and $G_2\subseteq\lceil G_2\rceil\cap Q$.
 By condition (CA2), we have some $F_1\in \mathcal{F}_{\tau}$ with $G_1\cup G_2\subseteq\lceil F_1\rceil\subseteq Q$. And by condition (CA1), there exists some $F_2\in  \mathcal{F}_{\tau}$ such that $G_1\cup G_2\subseteq\lceil F_2\rceil$ and $F_2\subseteq\lceil F_1\rceil\subseteq Q$. Then by condition (AD), $F_2\subseteq\lceil G_3\rceil$, $G_3\subseteq\lceil G_3\rceil$ and $G_3\subseteq\lceil F_1\rceil\subseteq Q$ for some $G_3\in \mathcal{F}_{\tau}$. Therefore, $\lceil G_3\rceil\in K(Q)$ and $\lceil G_1\rceil,\lceil G_2\rceil\subseteq\lceil G_3\rceil$. The remainder of the proof is to show that $Q=\bigcup K(Q)$. Since $\bigcup K(Q)\subseteq Q$ is clear, it suffices to show that  $Q\subseteq\bigcup K(Q)$. To this end, let $x\in Q$. Using condition (CA2), we have some $F_1\in\mathcal{F}_{\tau}$ with $x\in \lceil F_1\rceil\subseteq Q$. By condition (CA1), there exists some $F_2\in\mathcal{F}_{\tau}$ such that $x\in \lceil F_2\rceil$ and $F_2\subseteq \lceil F_1\rceil\subseteq Q$. According to condition (AD), there exists some $F\in\mathcal{F}_{\tau}$ such that $F_2\subseteq \lceil F\rceil$, $F\subseteq \lceil F\rceil$ and $F\subseteq \lceil F_1\rceil$. So  $x\in \lceil F\rceil$ and $\lceil F\rceil\in K(Q)$. Therefore, $x\in \bigcup K(Q)$ and hence $Q\subseteq \bigcup K(Q)$.
\end{proof}
In most applications of domain theory, it is frequently discussed whether a domain has a least element or a greatest element.
\begin{proposition}\label{p3.16}
Let $\mathbb{P}_{\tau}=(P_o,P_a,\vDash,\mathcal{F}_{\tau})$ be an attribute continuous formal context. Then
$(\mathfrak{B}(P_a,\mathcal{F}_{\tau}),\subseteq)$ is pointed if and only if,  for any $F\in \mathcal{F}_{\tau}$ there exists some  $G\in \mathcal{F}_{\tau}$ with $G\subseteq \lceil F\rceil$.
\end{proposition}
\begin{proof}
Suppose that $(\mathfrak{B}(P_a,\mathcal{F}_{\tau}),\subseteq)$ has a least element $\bot$, Then $\bot\ll\bot$. this means that there exists some $G\in \mathcal{F}_{\tau}$ such that $\bot\subseteq \lceil G\rceil$ and $G\subseteq \bot$. Therefore, $G\subseteq \lceil F\rceil$ for any $F\in \mathcal{F}_{\tau}$.

For the converse implication, assume that $\mathcal{F}_{\tau}$ has some element $G$ such that $G\subseteq \lceil F\rceil$ for any $F\in \mathcal{F}_{\tau}$. Then $\lceil G\rceil\subseteq \lceil F\rceil$. For any $Q\in \mathfrak{B}(P_a,\mathcal{F}_{\tau}$), we have known that $Q$ is the union of the set $\set{\lceil F\rceil}{F\in \mathcal{F}_{\tau}, F\subseteq Q}$. Thus $\lceil G\rceil\subseteq Q$, that is, $\lceil G\rceil$ is a least element of $(\mathfrak{B}(P_a,\mathcal{F}_{\tau}),\subseteq)$.
\end{proof}

\begin{proposition}\label{p3.17}
Let $\mathbb{P}_{\tau}=(P_o,P_a,\vDash,\mathcal{F}_{\tau})$ be an attribute continuous formal context.
Then
$(\mathfrak{B}(P_a,\mathcal{F}_{\tau}),\subseteq)$ has a greatest element if and only if $ \bigcup_{F\in\mathcal{F}_{\tau}}\lceil F\rceil$ is a continuous formal concept.
\end{proposition}
\begin{proof}
Assume that $\mathfrak{B}(P_a,\mathcal{F}_{\tau})$ has a greatest element $\top$. We show that $\bigcup_{F\in\mathcal{F}_{\tau}}\lceil F\rceil$ is a continuous formal concept by checking that $\bigcup_{F\in\mathcal{F}_{\tau}}\lceil F\rceil$ satisfies condition~(CA2). For this, let $M\sqsubseteq \bigcup_{F\in\mathcal{F}_{\tau}}\lceil F\rceil$. Then $M\sqsubseteq \top$. This implies that
  $M\subseteq \lceil F\rceil$ for some $F\in \mathcal{F}_{\tau}$, and hence
 $M\subseteq \lceil F\rceil\subseteq\bigcup_{F\in\mathcal{F}_{\tau}}\lceil F\rceil$.

For the converse implication, note that $Q=\bigcup\set{\lceil F\rceil}{F\in \mathcal{F}_{\tau}, F\subseteq Q}$ for any $Q\in \mathfrak{B}(P_a,\mathcal{F}_{\tau})$, it follows that $Q\subseteq \bigcup_{F\in\mathcal{F}_{\tau}}\lceil F\rceil$. Therefore, $\bigcup_{F\in\mathcal{F}_{\tau}}\lceil F\rceil$ is a greatest element of $\mathfrak{B}(P_a,\mathcal{F}_{\tau})$.
\end{proof}
\begin{definition}
Let $\mathbb{P}_{\tau}=(P_o,P_a,\vDash,\mathcal{F}_{\tau})$ be an attribute continuous formal context.
\begin{enumerate}[(1)]
\item $\mathbb{P}_{\tau}$ is called a dense attribute continuous formal context, if it satisfies condition~(AD).
\item $\mathbb{P}_{\tau}$ is called a pointed attribute continuous formal context if,  for any $F\in \mathcal{F}_{\tau}$ there exists some  $G\in \mathcal{F}_{\tau}$ with $G\subseteq \lceil F\rceil$.
\item $\mathbb{P}_{\tau}$ is called a topped attribute continuous formal context, if  $\bigcup_{F\in\mathcal{F}_{\tau}}\lceil F\rceil$ is a continuous formal concept.
    \end{enumerate}
\end{definition}
\begin{corollary}
Let $(D,\leq)$ be a continuous domain with a basis $B_D$ and
$Rep(D)$  the associated attribute continuous formal context.
\begin{enumerate}[(1)]
\item If $(D,\leq)$ is algebraic, then  $Rep(D)$ is a dense attribute continuous formal context.
\item If $(D,\leq)$ has a least element, then $Rep(D)$ is a pointed attribute continuous formal context.
\item If $(D,\leq)$ has a greatest element, then $Rep(D)$ is a topped attribute continuous formal context.
\end{enumerate}
\end{corollary}
\begin{proof}
(1) Use Theorem~\ref{t3.14} and Proposition~\ref{p3.15}.

(2) Use Theorem
~\ref{t3.14} and Proposition~\ref{p3.16}.

(3) Use Theorem
~\ref{t3.14} and Proposition~\ref{p3.17}.
\end{proof}
\begin{definition}
An attribute continuous formal context~$\mathbb{P}_{\tau}=(P_o,P_a,\vDash,\mathcal{F}_{\tau})$ is said to be a consistent attribute continuous formal formal context if it satisfies the following condition
\begin{enumerate}[{\bf(BC)}]
\item $(\forall F\in\mathcal{F}_{\tau})\emptyset\neq
X\sqsubseteq \lceil F\rceil\Rightarrow X\in \mathcal{F}_{\tau}$.
\end{enumerate}
\end{definition}

\begin{proposition}

(1) Let $\mathbb{P}_{\tau}=(P_o,P_a,\vDash,\mathcal{F}_{\tau})$ be a consistent attribute continuous formal  formal context. Then
$(\mathfrak{B}(P_{\tau},\mathcal{F}_{\tau}),\subseteq)$ is a  bounded complete domain.

(2) Let $(D,\leq)$ be a bounded complete domain with a basis~$B_D$. Then the associated attribute continuous formal context~$(D,B_D,\geq,\mathcal{F}_{\tau})$  is  consistent.
\end{proposition}
\begin{proof}

(1) To show $(\mathfrak{B}(P_{\tau},\mathcal{F}_{\tau}),\subseteq)$ is a  bounded complete domain, with Theorem~\ref{t3.14},
it suffices to show that any two elements of $\mathfrak{B}(P_{\tau},\mathcal{F}_{\tau})$ which bound above have a least upper bound. Let $Q_1, Q_2$ and $Q_3$ are elements of $\mathfrak{B}(P_{\tau},\mathcal{F}_{\tau})$ with $Q_1, Q_2\subseteq Q_3$. Put
$$Q=\set{x\in P_a}{(\exists F\in \mathcal{F}_{\tau})(F\subseteq Q_1\cup Q_2, x\in\lceil F\rceil)}.$$
We  show that $Q$ is also an elements of $\mathfrak{B}(P_{\tau},\mathcal{F}_{\tau})$ and it is the least upper bound of $Q_1$ and $Q_2$ in the following.

Let $M\subseteq Q$. If $M=\emptyset$, then $M\subseteq \lceil F\rceil\subseteq Q$ for any $F\in\mathcal{F}_{\tau}$ with $F\subseteq Q_1\cup Q_2$. Suppose now that $M\neq \emptyset$, then for any $x\in M$, there exists some $F_x\in \mathcal{F}_{\tau}$ such that $F_x\subseteq Q_1\cup Q_2$ and $x\in \lceil F_x\rceil$. Then $\bigcup_{x\in M}F_x \sqsubseteq Q_1\cup Q_2\subseteq Q_3$. By condition (BC), we have $\bigcup_{x\in M}F_x\in \mathcal{F}_{\tau}$. As a result, $M\subseteq \lceil\bigcup_{x\in M}F_x\rceil\subseteq Q_1\cup Q_2$. That is, $Q$ is an elements of $\mathfrak{B}(P_{\tau},\mathcal{F}_{\tau})$.

Since $Q_1, Q_2 \subseteq Q$ is clear, we have to show that $Q\subseteq Q_0$ for any $Q_0\in \mathfrak{B}(P_{\tau},\mathcal{F}_{\tau})$ with $Q_1,Q_2\subseteq Q_0$. For any $x\in Q$, there exists some $F\in\mathcal{F}_{\tau}$ such that $F\subseteq Q_1\cup Q_2$ and $x\in\lceil F\rceil$. Then $F\subseteq Q_0$. Note that $Q_0=\bigcup\set{\lceil F\rceil}{F\in\mathcal{F}_{\tau},F\subseteq Q_0}$, it follows that $Q\subseteq Q_0$.

(2) Let $(D,\leq)$ be a bounded complete domain with a basis~$B_D$. Suppose that $F\in\mathcal{F}_{\tau}$ with $\emptyset\neq X\sqsubseteq \lceil F\rceil$, where $\mathcal{F}_{\tau}$ is defined by  equation~\ref{e3.4}. This implies that $\bigvee F\in F$ and $X\leq \bigvee F$. Since $(D\leq)$ is a  bounded complete domain, $\bigvee X\in D$. Note that $X\neq \emptyset$ and $X$ is finite, it follows that $\bigvee X\in X$, which implies that $X\in \mathcal{F}_{\tau}$.
\end{proof}

With the above propositions stated in this subsection, the following corollary is obvious.
\begin{corollary}

 Let $\mathbb{P}_{\tau}=(P_o,P_a,\vDash,\mathcal{F}_{\tau})$ be a topped   consistent attribute continuous formal formal context. Then $(\mathfrak{B}(P_a,\mathcal{F}_{\tau}),\subseteq)$  is a continuous lattice. Moreover, if $\mathbb{P}_{\tau}$ is also dense, then $(\mathfrak{B}(P_a,\mathcal{F}_{\tau}),\subseteq)$ is an algebraic lattice.

 Conversely,
let $(D,\leq)$ be a continuous lattice with a basis~$B_D$. Then the associated attribute continuous formal context~$Rep(D)$  is a topped  consistent attribute continuous formal context. Moreover, if $(D,\leq)$ is an algebraic lattice, then $Rep(D)$  is a topped  dense consistent attribute continuous formal formal context.
\end{corollary}
  We finish this section by providing an approach  to representing stably continuous semilattices in the sense of FCA.
  \begin{definition}
  An attribute continuous formal context is said to be  multiplicative  if for any $F_1,F_2,G_1,G_2\in\mathcal{F}_{\tau}$,
\begin{enumerate}[{\bf(SS1)}]
\item $M\sqsubseteq \lceil F_1\rceil\cap \lceil F_2\rceil\Rightarrow (\exists F\in\mathcal{F}_{\tau})(F\subseteq \lceil F_1\rceil\cap \lceil F_2\rceil, M\subseteq\lceil F\rceil)$,
\item $G_1\subseteq \lceil F_1\rceil, G_2\subseteq \lceil F_2\rceil\Rightarrow (\exists F,G\in\mathcal{F}_{\tau})(\lceil G_1\rceil\cap \lceil G_2\rceil\subseteq \lceil G\rceil,G\subseteq \lceil F\rceil\subseteq \lceil F_1\rceil\cap \lceil F_2\rceil)$,
\end{enumerate}

  \end{definition}
\begin{proposition}~\label{p3.24}
(1) Let $\mathbb{P}_{\tau}=(P_o,P_a,\vDash,\mathcal{F}_{\tau})$ be a multiplicative  attribute continuous formal context. Then
$(\mathfrak{B}(P_a,\mathcal{F}_{\tau}),\subseteq)$ is a stably continuous semilattice.

(2) For any stably continuous semilattice $(D,\leq)$ with a basis~$B_D$, the associated attribute continuous formal context~$(D,B_D,\geq,\mathcal{F}_{\tau})$ is multiplicative.

\end{proposition}
\begin{proof} (1) To show that $(\mathfrak{B}(P_a,\mathcal{F}_{\tau}),\subseteq)$ is a stably continuous semilattice, by Theorem~\ref{t3.14}, we only need to prove that
it is a semilattice and the way-below relation on it is multiplicative.

Assume that $Q_1,Q_2\in\mathfrak{B}(P_a,\mathcal{F}_{\tau})$ and $M\sqsubseteq Q_1\cap Q_2$. According to condition~(CA2), we have some $F_1,F_2\in \mathcal{F}_{\tau}$ such that $M\subseteq \lceil F_1\rceil\subseteq Q_1$ and
$M\subseteq \lceil F_2\rceil\subseteq Q_2$. Then $M\subseteq \lceil F_1\rceil\cap\lceil F_2\rceil\subseteq Q_1\cap Q_2$. By condition (SS1), there exists some $F\in\mathcal{F}_{\tau}$ such that $F\subseteq \lceil F_1\rceil\cap\lceil F_2\rceil\subseteq  Q_1\cap Q_2$ and $M\subseteq \lceil F\rceil$. This implies that $Q_1\cap Q_2\in\mathfrak{B}(P_a,\mathcal{F}_{\tau})$ and hence $(\mathfrak{B}(P_a,\mathcal{F}_{\tau}),\subseteq)$ is a semilattice.

Assume that $Q_1,Q_2$ and $Q_3$ are elements of $\mathfrak{B}(P_a,\mathcal{F}_{\tau})$ such that $Q_3\ll Q_1$ and $Q_3\ll Q_2$. Then by equation~(\ref{e3.2}), there exist some $G_i\in \mathcal{F}_{\tau}$ such that $Q_3 \subseteq \lceil G_i\rceil$ and $G_i\subseteq Q_i$, where $i=1,2$. According to condition~(CA2), we have some $F_i\in\mathcal{F}_{\tau}$ with $G_i\subseteq \lceil F_i\rceil\subseteq Q_i$, $i=1,2$. By condition~(SS2), it follows that
$Q_3\subseteq\lceil G_1\rceil\cap \lceil G_2\rceil\subseteq \lceil G\rceil~\text{and}~ G\subseteq \lceil F\rceil\subseteq \lceil F_1\rceil\cap \lceil F_2\rceil\subseteq Q_1\cap Q_2,$
for some $ F,G\in\mathcal{F}_{\tau}$. As a result, $Q_3\ll Q_1\cap Q_2$.

(2)
Let $(D,\leq)$ be a continuous bounded complete domain with a basis~$B_D$. It suffices to show that $Rep(D)$ satisfies condition~(SS1) and (SS2).

For condition~(SS1), let $F_1,F_2\in\mathcal{F}_{\tau}$ with $M\sqsubseteq \lceil F_1\rceil\cap \lceil F_2\rceil$. Then $M\subseteq\dda(\bigvee F_1)\cap \dda(\bigvee F_2)\cap B_D$. By interpolation of $\ll$, we have $M\subseteq \dda x_1\cap \dda x_2$ for some $x_1\ll \bigvee F_1$ and $x_2\ll \bigvee F_2$. Since $\ll$ is multiplicative, $M\ll x_1\wedge x_2$. Take $F=\{x_1\wedge x_2\}$, thus $F\subseteq \lceil F_1\rceil\cap \lceil F_2\rceil$ and $ M\subseteq\lceil F\rceil$.

For conditions~(SS2),  assume that $G_1\subseteq \lceil F_1\rceil$ and $G_2\subseteq \lceil F_2\rceil$, where $F_1,F_2,G_1,G_2\in\mathcal{F}_{\tau}$.
Then $\bigvee G_1\ll \bigvee F_1$ and $\bigvee G_2\ll \bigvee F_2$, and hence $\bigvee G_1\wedge \bigvee G_2\ll \bigvee F_1\wedge \bigvee F_2$. Take $F=\{\bigvee F_1\wedge \bigvee F_2\}$ and $G=\{\bigvee G_1\wedge \bigvee G_2\}$, we have $ \lceil G_1\rceil\cap \lceil G_2\rceil\subseteq \lceil G\rceil$ and $G\subseteq \lceil F\rceil\subseteq \lceil F_1\rceil\cap \lceil F_2\rceil)$.
\end{proof}

\section{The category of attribute continuous formal contexts}
   From a categorical view of  point, the previous section has provided object part correspondence between continuous domains and  attribute continuous formal contexts. In this section, we extend this relationship to a categorical equivalence.
 \subsection{$\mathcal{F}$-morphisms}
     On the side of continuous domains, Scott-continuous functions is typically used as the morphisms to form a category.
     In this subsection, we identify a notion of $\mathcal{F}$-morphism between continuous formal attribute  contexts which can be used to represent the Scott-continuous functions between continuous domains. Similar to the case of approximable concepts, the morphisms we defined are relations instead of mappings.

\begin{definition}
 Consider attribute continuous formal contexts $\mathbb{P}_{\tau}=(P_o,P_a,\vDash, \mathcal{F}_{\tau})$ and $\mathbb{P}'_{\tau'}=(P'_o,P'_a,\vDash', \mathcal{F'}_{\tau'})$. An $\mathcal{F}$-morphism from $\mathbb{P}_{\tau}$ to $\mathbb{P}'_{\tau'}$ is a relation~$\mathbb{H}\subseteq \mathcal{F}_{\tau}\times P_a'$, written as $\mathbb{H}:\mathbb{P}_{\tau}\rightarrow\mathbb{P}'_{\tau'}$, such that for any $F,G\in\mathcal{F}_{\tau}$, $F'\sqsubseteq \mathcal{F}'_{\tau'}$ and $M'\sqsubseteq P_a'$, the following conditions hold,
\begin{enumerate}[{\bf({AR}1)}]
\item $(F\mathbb{H} F',x'\in \lceil F'\rceil)\Rightarrow (F,x')\in\mathbb{H}$,
 \item $(G\subseteq \lceil F\rceil, (G,x')\in\mathbb{H})\Rightarrow (F,x')\in\mathbb{H}$,
 \item $F\mathbb{H}M'\Rightarrow (\exists F'\in\mathcal{F}'_{\tau'})(M'\subseteq F',F\mathbb{H}F')$,
 \item $(F,x')\in\mathbb{H}\Rightarrow(\exists F_x\in \mathcal{F}_{\tau})(\exists F_x'\in\mathcal{F}'_{\tau'})(F_x\subseteq \lceil F\rceil, x'\in \lceil F_x'\rceil, F_x\mathbb{H} F_x')$.
 \end{enumerate}
where $F\mathbb{H} X'$ means that $(F,x')\in\mathbb{H}$ for any $x'\in X'$.
\end{definition}

 Given attribute continuous formal context $\mathbb{P}_{\tau}=(P_o,P_a,\vDash, \mathcal{F}_{\tau})$, define a relation $\mathbb{H}\subseteq \mathcal{F}_{\tau}\times P_a$ by $(F,x)\in \mathbb{H}$ if and only if $x\in \lceil F\rceil$, then $\mathbb{H}$ is a special $\mathcal{F}$-morphism on $\mathbb{P}_{\tau}$.

\begin{proposition}\label{p4.2}
For attribute continuous formal contexts $\mathbb{P}_{\tau}=(P_o,P_a,\vDash, \mathcal{F}_{\tau})$ and $\mathbb{P}'_{\tau'}=(P'_o,P'_a,\vDash', \mathcal{F'}_{\tau'})$,
if a relation~$\mathbb{H}\subseteq \mathcal{F}_{\tau}\times P_a'$ satisfies conditions (AR1) and~(AR2), then conditions (AR3) and~(AR4) together are equivalent to the following one:
\begin{enumerate}[\bf{({AR}5)}]
\item $F\mathbb{H} M'\Rightarrow (\exists G\in \mathcal{F}_{\tau})(\exists G'\in\mathcal{F}'_{\tau})(G\subseteq \lceil F\rceil, M'\subseteq \lceil G'\rceil, G\mathbb{H} G')$.
 \end{enumerate}
\end{proposition}
\begin{proof}
Suppose that $\mathbb{H}\subseteq \mathcal{F}_{\tau}\times P_a'$ satisfies conditions (AR1---AR4). Let $F\in\mathcal{F}_{\tau}$ and $M'\sqsubseteq P'_{a}$ with $F\mathbb{H} M'$. Then by condition~(AR3), there exists some $F'\in\mathcal{F}'_{\tau'}$ such that $M'\subseteq F'$ and $F\mathbb{H} F'$.  Then for any $x'\in F'$, by condition~(AR4), there exist $F_{x'}\in \mathcal{F}_{\tau}$ and $F'_{x'}\in \mathcal{F}'_{\tau'}$ such that $F_{x'}\subseteq \lceil F\rceil$, $F_{x'}\mathbb{H}F'_{x'}$ and $x'\in \lceil F'_{x'}\rceil$.  Since $\bigcup_{x'\in M'}F_{x'}\sqsubseteq \lceil F\rceil$, by condition~(CA1), there exists $G\in\mathcal{F}_{\tau}$ such that $\bigcup_{x'\in M'}G_{x'}\sqsubseteq \lceil G\rceil$ and $G\sqsubseteq \lceil F\rceil$. According to condition~(AR2), we have $G\mathbb{H}\bigcup_{x'\in M'}G'_{x'}$. By condition (AR3), there exists $G'\in \mathcal{F}'_{\tau'}$ such that $\bigcup_{x'\in M'}G_{x'}\sqsubseteq G'$ and $G\mathbb{H}G'$. As has already mentioned, $\bigcup_{x'\in M'}F_{x'}\sqsubseteq G'$ and $x'\in \lceil F_{x'}\rceil$ for any $x'\in M$. So $M'\subseteq \lceil G'\rceil$ holds. In conclusion, condition~(AR5) follows.

Conversely, suppose that $\mathbb{H}\subseteq \mathcal{F}_{\tau}\times P_a'$ satisfies conditions (AR1), (AR4) and (AR5). Since condition (AR4) is a special case of (AR5),  it suffices to prove that $\mathbb{H}$ satisfies condition~(AR3). For this, suppose that $F\mathbb{H} M'$. Then by condition~(AR5), there exists $G\in\mathcal{F}_{\tau} $ and $G'\in \mathcal{F'}_{\tau}$ such that $G \subseteq\lceil F\rceil$, $G\mathbb{H} G'$ and $M'\subseteq \lceil G'\rceil$. For $G \subseteq\lceil F\rceil$ and $G\mathbb{H} G'$, with condition~(AR2), it follows that $F\mathbb{H} G'$.
\end{proof}

In fact, Proposition~\ref{p4.2} provides an alternative description of $\mathcal{F}$-morphism $\mathbb{H}: \mathbb{P}_{\tau}\rightarrow\mathbb{P'}_{\tau'}$. The following basic properties of $\mathcal{F}$-morphisms will be often used in the subsequent section.
\begin{proposition}\label{p4.3}
Let  $\mathbb{H}: \mathbb{P}_{\tau}\rightarrow\mathbb{P'}_{\tau'}$ be an $\mathcal{F}$-morphism. If $F,F_1 \in\mathcal{F}'_{\tau'}$ and $M'\sqsubseteq P'_{a}$, then the following statements hold.
\begin{enumerate}[(1)]
\item  $F\mathbb{H}M'$ if and only if there exists some $G\subseteq \lceil F\rceil$ such that $G\mathbb{H}M'$.
\item If $F_1\subseteq F$ and $F_1\mathbb{H}M'$, then $F\mathbb{H}M'$.
 \end{enumerate}
\end{proposition}
\begin{proof}
(1) Suppose that $F\mathbb{H} M'$. Then by condition~(AR5), there exist $G\in\mathcal{F}_{\tau} $ and $G'\in \mathcal{F'}_{\tau}$ such that $G \subseteq\lceil F\rceil $, $G\mathbb{H} G'$ and $M'\subseteq \lceil G'\rceil $. For $G\mathbb{H} G'$, using condition~(AR1), we have  $G\mathbb{H} \lceil  G'\rceil $. From $M'\subseteq \lceil  G'\rceil $, it follows that $G\mathbb{H} M'$. Conversely, if $G\mathbb{H}M'$ for some $G\subseteq \lceil F\rceil$, then with condition~(AR2), it is obvious that  $F\mathbb{H}M'$.

(2) Assume that $F_1\subseteq F$ and $F_1\mathbb{H} M'$. For $F_1\mathbb{H} M'$, using condition~(AR5), we have $G\in\mathcal{F}_{\tau} $ and $G'\in \mathcal{F'}_{\tau}$ such that $G \subseteq\lceil F_1\rceil $, $G\mathbb{H} G'$ and $M'\subseteq \lceil G'\rceil $. Then $G\mathbb{H} M'$. Because $G \subseteq\lceil F_1\rceil $ and $F_1\subseteq F$, it follows that $G \subseteq\lceil F\rceil $. By condition~(AR2), we have $F\mathbb{H} M'$.
\end{proof}

Let  $\mathbb{H}: \mathbb{P}_{\tau}\rightarrow\mathbb{P'}_{\tau'}$ be an $\mathcal{F}$-morphism. For any subset $X$ of $P_{a}$,  define
\begin{equation}\label{e4.1}
\mathbb{H}(X)=\set{x'\in P'_{a}}{(\exists F\in\mathcal{F}_{\tau})(F\subseteq X,(F,x')\in\mathbb{H})}.
\end{equation}

The next proposition shows that the notion of $\mathcal{F}$-morphisms builds a passage from
continuous formal attribute concepts of an attribute continuous formal context to those of another one.
\begin{proposition}\label{p4.4}
Consider  an $\mathcal{F}$-morphism $\mathbb{H}: \mathbb{P}_{\tau}\rightarrow\mathbb{P'}_{\tau'}$.
\begin{enumerate}[(1)]
\item If $F\in \mathcal{F}_{\tau}$, then $\mathbb{H}(F)$ is a continuous formal attribute concept of $\mathbb{P'}_{\tau'}$, and hence $\mathbb{H}(F)\neq \emptyset$.
    \item  If $Q$ is a continuous formal concept of $\mathbb{P}_{\tau}$, then $\mathbb{H}(Q)$ is a continuous formal concept of $\mathbb{P'}_{\tau'}$.
 \end{enumerate}
\end{proposition}
\begin{proof}
(1) Suppose that  $F\in \mathcal{F}_{\tau}$. By part~(2) of Proposition~4.3 and equation 4.1, we have
\begin{equation}
\mathbb{H}(F)=\set{x'\in P'_{a}}{(F,x')\in\mathbb{H}}.
\end{equation}
  Moreover, we claim that $$\mathbb{H}(F)=\bigcup \set{\lceil F'\rceil}{F'\in \mathcal{F}_{\tau'}',F\mathbb{H} F'}.$$
  In fact, for any $x'\in \mathbb{H}(F)$, that is, $F\mathbb{H} x'$. By condition~(AR3), there exists $F'\in \mathcal{F}'_{\tau'}$ such that $\{x'\}\subseteq \lceil F'\rceil$ and $F\mathbb{H} F'$. According to condition~(AR1), it follows that $F\mathbb{H} \lceil F'\rceil$. So
  $x'\in \bigcup \set{\lceil F'\rceil}{F'\in \mathcal{F}_{\tau'}',F\mathbb{H} F'}$, and thus
  $ \mathbb{H}(F)\subseteq\bigcup \set{\lceil F'\rceil}{F'\in \mathcal{F}_{\tau'}',F\mathbb{H} F'}$. For the reverse inclusion, let $x'\in \bigcup \set{\lceil F'\rceil}{F'\in \mathcal{F}_{\tau'}',F\mathbb{H} F'}$. Then $x'\in \lceil F'\rceil $ for some $ F'\in \mathcal{F'}_{\tau}$ with $F\mathbb{H} F'$. From condition (AR1), it follows that $F\mathbb{H} \lceil F'\rceil $. Hence $x'\in \mathbb{H}(F)$, which implies that $\bigcup \set{\lceil F'\rceil}{F'\in \mathcal{F}_{\tau'}',F\mathbb{H} F'}\subseteq \mathbb{H}(F)$.

 We now prove that the set
 $\set{\lceil F'\rceil}{F'\in \mathcal{F}_{\tau'}',F\mathbb{H} F'}$ is directed. For this, let $F'_1,F'_2\in \mathcal{F}_{\tau'}$ such that  $F\mathbb{H} F_1'$ and $F\mathbb{H} F_2'$. Then $F\mathbb{H}( F_1'\cup F_2')$. By condition~(AR3), there exists $G'\in\mathcal{F'}$ such that $F_1'\cup F_2' \subseteq \lceil G'\rceil$  and $F\mathbb{H} G'$. As $G'\in \mathcal{F}'_{\tau'}$, by condition~(CA1), there exists $F_3'\in \mathcal{F}'_{\tau'}$ such that $F_1'\cup F_2' \subseteq \lceil F_3'\rceil$ and $F_3'\subseteq \lceil G'\rceil$. Note that $F_1'\cup F_2' \subseteq \lceil F_3'\rceil$, it follows that $\lceil F_1'\rceil \subseteq \lceil F_3'\rceil$ and $\lceil F_2'\rceil \subseteq \lceil F_3'\rceil$. From $F_3'\subseteq \lceil G'\rceil$ and $F\mathbb{H} G'$, by condition~(AR1), we have $F\mathbb{H} F_3$.  In summery, $\set{\lceil F'\rceil}{F'\in \mathcal{F}_{\tau'}',F\mathbb{H} F'}$ is directed.

 Proposition~\ref{p3.10} has proven that $\lceil F'\rceil$ is a continuous formal attribute concept of $\mathbb{P'}_{\tau'}$ for any $F'\in \mathcal{F}_{\tau'}'$.  Then by Proposition~\ref{p3.12}, $\mathbb{H}(F)=\bigcup \set{\lceil F'\rceil}{F'\in \mathcal{F}_{\tau'}',F\mathbb{H} F'}$ is  a continuous formal attribute concept of $\mathbb{P'}_{\tau'}$.

 (2) We show that $\mathbb{H}(Q)$ is a continuous formal concept of $\mathbb{P'}_{\tau'}$ by checking  $\mathbb{H}(Q)$ satisfies condition~(CA2). Let $M'\subseteq\mathbb{H}(Q)$. The subsequent reasoning is divided into two cases. Case~(i): $M'=\emptyset$. By Proposition~\ref{p3.11}, there exists $F\in\mathcal{F}_{\tau}$ such that $F\subseteq Q$, which implies that $\emptyset\neq \mathbb{H}(F)\subseteq \mathbb{H}(Q)$. Taking $x'\in \mathbb{H}(Q)$ ,we have some $F_{x'}\in\mathcal{F}_{\tau}$ such that $F_{x'}\subseteq Q$ and $(F_{x'},x)\in\mathbb{H}$. Using conditions~(AR2) and (AR1), there exists $F'\in\mathcal{F}'_{\tau'}$ satisfying $F_{x'}\mathbb{H}\lceil F'\rceil$. So $M'\subseteq\lceil F'\rceil\subseteq \mathbb{H}(Q)$. Case~(ii): $M'\neq \emptyset$. For any $m'\in M'$, there exists some $F_{m'}\in\mathcal{F}_{\tau}$ such that $F_{m'}\subseteq Q$ and $F_{m'}\mathbb{H} m'$. Since $\bigcup_{m'\in M'}F_{m'}\sqsubseteq Q$, we have some $F\in\mathcal{F}_{\tau}$ such that $F\subseteq Q$ and $\bigcup_{m'\in M'}F_{m'}\subseteq \lceil F\rceil$.  Thus by condition~(AR2), it follows that $F\mathbb{H} M'$.
 This implies that  $M'\subseteq \lceil G'\rceil$ and $F\mathbb{H} G'$ for some $G'\in \mathcal{F}_{\tau'}'$. Hence $M'\subseteq \lceil G'\rceil\subseteq \mathbb{H}(Q)$, as required.
\end{proof}

Now we turn to discuss how Scott-continuous functions between continuous domains can be captured by $\mathcal{F}$-morphisms between attribute continuous formal contexts.
 Let us start with attribute continuous formal contexts $\mathbb{P}_{\tau}=(P_o,P_a,\vDash, \mathcal{F}_{\tau})$ and $\mathbb{P}'_{\tau'}=(P'_o,P'_a,\vDash', \mathcal{F'}_{\tau'})$. And we first show that there is
  a one-to-one correspondence between $\mathcal{F}$-morphisms from $\mathbb{P}_{\tau}$ to $\mathbb{P}'_{\tau'}$ and Scott-continuous functions from $\mathfrak{B}((P_a,\mathcal{F}_{\tau}),\subseteq)$ to $\mathfrak{B}((P_a',\mathcal{F}'_{\tau'}),\subseteq)$.

\begin{theorem} \label{t4.5}
Let $\mathbb{P}_{\tau}$ and $\mathbb{P}'_{\tau'}$ be two attribute continuous formal contexts.
\begin{enumerate}[(1)]
 \item For any  $\mathcal{F}$-morphism~$\mathbb{H}: \mathbb{P}_{\tau}\rightarrow\mathbb{P'}_{\tau'}$, define a function
$\phi_{\mathbb{H}}:\mathfrak{B}(P_a,\mathcal{F}_{\tau})\rightarrow \mathfrak{B}(P_a,\mathcal{F}'_{\tau'})$ by
\begin{equation}\label{e4.3}
\phi_{\mathbb{H}}(Q)=\set{x'\in P'_{a}}{(\exists F\in\mathcal{F}_{\tau})(F\subseteq Q, (F,x')\in\mathbb{H})}.
\end{equation}
Then $\phi_{\mathbb{H}}$ is Scott-continuous.

\item For any Scott-continuous function $\phi : \mathfrak{B}(P_a,\mathcal{F}_{\tau})\rightarrow
    \mathfrak{B}(P_a',\mathcal{F}'_{\tau'})$, define
$\mathbb{H}_{\phi}\subseteq\mathcal{F}_{\tau}\times P'_{a}$ by
 \begin{equation}\label{e4.4}
(F,x')\in\mathbb{H}_{\phi}\Leftrightarrow x'\in \phi(\lceil F\rceil).
\end{equation}
Then $\mathbb{H}_{\phi}$ is an $\mathcal{F}$-morphism from $\mathbb{P}_{\tau}$ to $\mathbb{P}'_{\tau'}$.

\item $\mathbb{H}_{\phi_{\mathbb{H}}}=\mathbb{H}$ and  $\phi_{\mathbb{H}_{\phi}}=\phi$.
\end{enumerate}
\end{theorem}
\begin{proof}
 (1) Let $\mathbb{H}: \mathbb{P}_{\tau}\rightarrow\mathbb{P'}_{\tau'}$ be an $\mathcal{F}$-morphism.
 From part~(2) of Proposition~\ref{p4.3}, the function $\phi_{\mathbb{H}}$ is well-defined. With equation~(\ref{e4.3}), it is clear that $\phi_{\mathbb{H}}$ is order-preserving. Then  $\set{\phi_{\mathbb{H}}(Q_i)}{i\in I}$ is a directed subset  of $\mathfrak{B}(P'_a,\mathcal{F}_{\tau'}')$ for any directed subset $\set{Q_i}{i\in I}$ of $\mathfrak{B}(P_a,\mathcal{F}_{\tau})$. By part~(1) of Lemma~\ref{l3.13}, we know that $\bigvee_{i\in I}Q_i=\bigcup_{i\in I}Q_i$ and $\bigvee_{i\in I}\phi_{\mathbb{H}}(Q_i)=\bigcup_{i\in I}\phi_{\mathbb{H}}(Q_i)$.
   Now we show that $\phi_{\mathbb{H}}$ is Scott-continuous by checking that $\phi_{\mathbb{H}}(\bigcup_{i\in I}Q_i)=\bigcup_{i\in I}\phi_{\mathbb{H}}(Q_i)$.
Since $\bigcup_{i\in I}\phi_{\mathbb{H}}(Q_i)\subseteq \phi_{\mathbb{H}}(\bigcup_{i\in I}Q_i)$ is clear, we need only to
 prove the reverse inclusion holds. Suppose that $x'\in \phi_{\mathbb{H}}(\bigcup_{i\in I}Q_i)$, then there exists some $F\in \mathcal{F}_{\tau}$ such that $F\sqsubseteq \bigcup_{i\in I}Q_i$ and $(F,x')\in\mathbb{H}$. This implies that $F\sqsubseteq Q_j$ for some $j\in I$. Thus $x'\in  \phi_{\mathbb{H}}(Q_j)\subseteq\bigcup_{i\in I}\phi_{\mathbb{H}}(Q_i)$, and hence $\phi_{\Theta}(\bigcup_{i\in I}Q_i)\subseteq\bigcup_{i\in I}\phi_{\mathbb{H}}(Q_i)$.

(2) Suppose that $\phi$ is  a Scott-continuous function from $\mathfrak{B}(P_a,\mathcal{F}_{\tau})$ to $\mathfrak{B}(P_a',\mathcal{F}'_{\tau'})$. We show that the relation $\mathbb{H}_{\phi}$ is an $\mathcal{F}$-morphism from $\mathbb{P}_{\tau}$ to $\mathbb{P}'_{\tau'}$ by checking that $\mathbb{H}_{\phi}$ satisfies conditions~(AR1), (AR2) and (AR5). Take $F, G\in\mathcal{F}_{\tau}$, $F'\sqsubseteq \mathcal{F}'_{\tau'}$ and $M'\sqsubseteq P_a$.

For condition~(AR1), if $F\mathbb{H}_{\phi}F'$ and $x'\in \lceil F'\rceil$, then $F'\subseteq \phi(\lceil F\rceil)$. Since $\lceil F\rceil\in \mathfrak{B}(P_a,\mathcal{F}_{\tau})$, by part~(2) of Proposition~\ref{p4.4}, it follows that $\phi(\lceil F\rceil)\in \mathfrak{B}(P_a',\mathcal{F}_{\tau'}')$. This implies that $\lceil F'\rceil \subseteq \phi(\lceil F\rceil)$. Thus $F\mathbb{H}_{\phi} \lceil F'\rceil $, and hence $(F,x')\in \mathbb{H}_{\phi}$.

For condition~(AR2), if $G\subseteq \lceil F\rceil$ and $(G,x')\in\mathbb{H}_{\phi}$, then $\lceil G\rceil\subseteq\lceil F\rceil$ and $x'\in \phi(\lceil G\rceil)$. This implies that $x'\in \phi (\lceil F\rceil)$, that is, $(F,x')\in \mathbb{H}_{\phi}$.

 For condition~(AR5), if $F \mathbb{H}_{\phi}M'$, then $M'\sqsubseteq\phi(\lceil F\rceil)$. Since $\lceil F\rceil$ is the directed union of the set $\set{\lceil G\rceil}{G\in\mathcal{F}_{\tau},G\sqsubseteq \lceil F\rceil}$ and $\phi$ is Scott-continuous, we have  $$\phi(\lceil F\rceil)=\phi(\bigcup \set{\lceil G\rceil}{G\in\mathcal{F}_{\tau},G\sqsubseteq \lceil F\rceil})=\bigcup \set{\phi(\lceil G\rceil)}{G\in\mathcal{F}_{\tau},G\sqsubseteq \lceil F\rceil}.$$ Thus $M'\sqsubseteq \phi(\lceil G\rceil)$ for some $G\in\mathcal{F}_{\tau}$ with $G\subseteq \lceil F\rceil$. To $M'\sqsubseteq \phi(\lceil G\rceil)$, using part~(2) of Proposition~\ref{l3.13}, we have some $G'\in \mathcal{F}_{\tau'}'$  satisfying $G'\subseteq \phi(\lceil G\rceil)$ and $M'\subseteq \lceil G'\rceil$. To sum up, there exist some $G\in\mathcal{F}_{\tau}$
and $G'\in \mathcal{F}_{\tau'}'$ such that $G\subseteq \lceil F\rceil$,  $M'\subseteq\lceil G'\rceil$ and $G\mathbb{H}_{\phi}G'$.

(3) For any $F\sqsubseteq X$ and $x'\in P_a'$, we have
\begin{align*}
(F,x')\in \mathbb{H}_{\phi_{\mathbb{H}}}&\Leftrightarrow x'\in \phi_{\mathbb{H}}(\lceil F\rceil)\\
&\Leftrightarrow (\exists G\in \mathcal{F}_{\tau})(G\subseteq \lceil F\rceil, (G,x')\in\mathbb{H})\\
&\Leftrightarrow (F,x')\in\mathbb{H}).
\end{align*}
This proves that $\mathbb{H}_{\phi_{\mathbb{H}}}=\mathbb{H}$.

For any $Q \in \mathfrak{B}(P_a,\mathcal{F}_{\tau})$, we have
\begin{align*}
\phi_{\mathbb{H}_{\phi}}(Q)&=\set{x'\in P_a'}{(\exists F\in\mathcal{F}_{\tau})(F\sqsubseteq Q,(F,x')\in\mathbb{H}_{\phi})}\\
&=\set{x'\in P_a'}{(\exists F\in\mathcal{F}_{\tau})(F\sqsubseteq Q,x\in\phi(\lceil F\rceil))}\\
&=\bigcup \set{\phi(\lceil F\rceil)}{(\exists F\in\mathcal{F}_{\tau})F\sqsubseteq Q}\\
&=\phi(\bigcup\set{\lceil F\rceil}{(\exists F\in\mathcal{F}_{\tau})F\sqsubseteq Q})\\
&=\phi(Q).
\end{align*}
This proves that $\phi_{\mathbb{H}_{\phi}}=\phi$.
\end{proof}
 Next, we investigate the relationship between Scott-continuous functions from $(D,\leq)$ to $(D',\leq')$ and
 $\mathcal{F}$-morphisms from $Rep(D)$ to  $Rep(D')$. Before this, we need the following two Lemmas.

 \begin{lemma}\label{l4.6} Let $(D,\leq)$ be a continuous domain with a basis $B_D$, and $(D',\leq')$  a continuous domain with a basis $B_{D'}$.
 If $f: D\rightarrow D'$ is a  Scott continuous function, then for any $x\in D$, we have
 \begin{equation}\label{e4.5}
 \dda f(x) \cap B_{D'}=\set{x'\in B_{D'}}{(\exists y\in B_{D})(y\ll x,x'\ll'f(y))}.
 \end{equation}
 \end{lemma}
 \begin{proof}
 For any $x\in D$ and $x'\in B_{D'}$, since
 \begin{align*}
x'\ll' \dda f(x)&\Leftrightarrow x'\ll' f(\bigvee(\dda x\cap B_D))\\
&\Leftrightarrow x'\ll' \bigvee f(\dda x\cap B_D)\\
&\Leftrightarrow (\exists y\in B_{D})(y\ll x,x'\ll'f(y)).
\end{align*}
The last equivalence follows from the definition of $\ll'$ and the monotonicity of $f$.
\end{proof}
 \begin{lemma}\label{l4.7}
 Let $(D,\leq)$ be a continuous domain with a basis $B_D$ and $(D',\leq')$  a continuous domain with a basis $B_{D'}$.
  For any $\mathcal{F}$-morphism $\mathbb{G}$ from  $Rep(D)=(D,B_D,\geq,\mathcal{F}_{\tau})$ to  $Rep(D)=(D',B_D',\geq',\mathcal{F}_{\tau'})$ and $x\in D$, put
 \begin{equation}\label{e4.6}
  I_x= \set{x'\in L' }{ (\exists F\in \mathcal{F}_\tau)(F\subseteq \dda x\cap B_D, (F,x')\in \mathbb{G})}.
 \end{equation}
 Then  $I_x$ has a least upper bound~$\bigvee I_x $ in $D'$.
\end{lemma}
\begin{proof}
For any $x\in D$, as has been stated  $\dda x\cap B_D$ is a
 continuous formal concepts of $Rep(D)$. With equations~(\ref{e4.1}) and (\ref{e4.6}), it is clear that  $I_x=\mathbb{G}(\dda x\cap B_D)$. Then from part~(2) of Proposition~\ref{p4.4}, it follows that $I_x$ is a
 continuous formal concepts of $Rep(D')$. This implies that $I_x$ is a directed subset of $D'$, and hence $\bigvee I_x \in D'$.
\end{proof}
The following result tells us that there is a one-to-one correspondence between Scott-continuous functions from $(D,\leq)$ to $(D',\leq')$ and
 $\mathcal{F}$-morphisms from $Rep(D)$ to  $Rep(D')$.
\begin{theorem}
Let $(D,\leq)$ be a continuous domain with a basis $B_D$, and $(D',\leq')$  a continuous domain with a basis $B_{D'}$.
\begin{enumerate}[(1)]
 \item
 Consider a Scott continuous function $f: D\rightarrow D'$. For any $F\sqsubseteq B_D$ with $\bigvee F\in F$ and $x'\in B_{D'}$, define a relation $\mathbb{G}_f$
 by
 \begin{equation}\label{e4.7}
 (F,x')\in\mathbb{G}_{f}\Leftrightarrow x'\ll' f(\bigvee F).
 \end{equation}
 Then $\mathbb{G}_{f}$ is an $\mathcal{F}$-morphism from $Rep(D)$ to  $Rep(D')$.
  \item  For any $\mathcal{F}$-morphism $\mathbb{G}$ from $Rep(D)$ to  $Rep(D')$ and $x\in D$, define
     \begin{equation}\label{e4.8}
       f_{\mathbb{G}}(x)=\bigvee \set{x'\in D'}{(\exists F\in \mathcal{F}_{\tau})(F\subseteq \dda x\cap B_{D}, (F,x')\in \mathbb{G})}.
       \end{equation}
  Then $f_{\mathbb{G}}:D\rightarrow D'$ is a Scott-continuous function.
   \item Moreover $f=f_{\mathbb{G}_{f}}$ and $\mathbb{G}= \mathbb{G}_{f_{\mathbb{G}}}$.
  \end{enumerate}
\end{theorem}
\begin{proof}
(1) For any Scott-continuous function $f : D \rightarrow D'$, we check that the relation $\mathbb{G}_f$ is an $\mathcal{F}$-morphism  from $Rep(D)$ to  $Rep(D')$ in the following:

For condition (AR1), if $F'\in \mathcal{F}_{\tau'}'$ and $F\mathbb{G}_f F'$, then $F'\ll' f(\bigvee F)$. Note that $\bigvee F'\in F'$, it follows that $\bigvee F'\ll'f(\bigvee F)$. This implies that $x'\ll'f(\bigvee F)$ and $(F,x')\in \mathbb{G}_f$ for any $x'\in\dda \bigvee F=\lceil F\rceil$.

For condition (AR2), suppose that $F\in\mathcal{F}_{\tau}$, $G\subseteq \lceil F\rceil$ and $(G,x')\in \mathbb{G}_f$. Then $G\ll \bigvee F$ and $x'\ll'f(\bigvee G)$. As $f$ is order-preserving, we have $x'\ll'f(\bigvee F)$. This means that  $(F,x')\in \mathbb{G}_f$.

 For condition (AR5), suppose that $F\mathbb{G}_f M'$ with $M'\sqsubseteq D'$. Then $M'\ll' f(\bigvee F)$. By the interpolation property of $\ll'$, there exists some $d'\in B_{D'}$ such that $M'\ll' d'\ll' f(\bigvee F)$.  Note that $f(\bigvee F)=f(\bigvee (\dda (\bigvee F)))=\bigvee (\dda f(\bigvee F))$, we have some $d\in \dda (\bigvee F)\cap B_D$ with $d'\ll' f(d)$. Set $G=\{d\}$ and $G'=\{d'\}\cup M'$. Thus $G\in \mathcal{F}_{\tau}$ and $G'\in \mathcal{F}_{\tau'}'$ such that $G\sqsubseteq \lceil F\rceil$, $M'\subseteq \lceil G'\rceil$ and $G\mathbb{G}_f G'$.

(2) Given an $\mathcal{F}$-morphism $\mathbb{G}$ from $Rep(D)$ to  $Rep(D)$. By Lemma~\ref{l4.7}, the function~$f_{\mathbb{G}}$ defined is well-defined.  Comparing equation~(\ref{e4.6}) and \ref{e4.8}), it is easy to see that $f_{\mathbb{G}}(x)=\bigvee I_x$, for any $x\in D$. We now prove that $f_{\mathbb{G}}$ is Scott-continuous by checking that $f_{\mathbb{G}}(\bigvee S)=\bigvee f_{\mathbb{G}} (S)$ for any directed subset $S$ of $D$.

   It is clear that $I_x\subseteq I_y$ for any $x,y \in D$ with $x\leq y$, which means that $f_{\mathbb{G}}$ is order-preserving. Therefore, $\bigvee f_{\mathbb{G}} (S)\leq f_{\mathbb{G}}(\bigvee S)$.
   Conversely, trivial checks verify that $\bigvee f_{\mathbb{G}} (S)= \bigvee \set{\bigvee I_d}{d\in S}=\bigvee (\bigcup{_{d\in S}}I_d).$
    Since $f_{\mathbb{G}}(\bigvee S)=\bigvee I_{\bigvee S}$, so that to complete the proof, it suffices to show that $I_{\bigvee S}\subseteq\bigcup{_{d\in S}}I_d$. If $x'\in  I_{\bigvee S}$, then $(F,x')\in\mathbb{G} $ for some $F\in \mathcal{F}_{\tau}$ with $F\sqsubseteq \dda \bigvee S\cap B_D$. Because $\bigvee F\in F$, we have $\bigvee F\ll\bigvee S$. By the definition of continuous domain,  there exists some $d\in S$ such that $\bigvee F\ll' d$. Thus $F\sqsubseteq \dda d \cap B_D$. This implies that $x'\in I_d$,
      and hence $I_{\bigvee S}\subseteq\bigcup{_{d\in S}}I_d$.

 (3) For any $x\in D$, since
 \begin{align*}
f_{\mathbb{G}_{f}}(x)&=\bigvee\set{x'\in B_{D'}}{(\exists F\in \mathcal{F}_{\tau})(F\sqsubseteq \dda x\cap B_D, (F,x')\in\mathbb{G}_{f}})\\
&=\bigvee\set{x'\in B_{D'}}{(\exists F\in \mathcal{F}_{\tau})(F\sqsubseteq \dda x\cap B_D,  x'\ll'f(\bigvee F)}\\
&=\bigvee \set{x'\in B_{D'}}{(\exists y\in B_D)( y\ll x, x'\ll'f(y))}\\
&=\bigvee (\dda f(x)\cap B_{D'})\\
&=f(x),
\end{align*}
 it follows  that $f=f_{\mathbb{G}_{f}}$.

 For any $F\sqsubseteq \mathcal{F}_{\tau}$ and $x'\in D'$, since
\begin{align*}
(F,x')\in\mathbb{G}_{f_{\mathbb{G}}}&\Leftrightarrow x'\ll'f_{\mathbb{G}}(\bigvee F)\\
&\Leftrightarrow x'\ll'\bigvee\set{y'\in D'}{(\exists F_1\in\mathcal{F}_{\tau})(F_1\subseteq \dda (\bigvee F)\cap B_D, (F,y')\in\mathbb{G}})\\
&\Leftrightarrow (\exists y'\in D', \exists F_1\in\mathcal{F}_{\tau})(F_1\sqsubseteq \dda (\bigvee F)\cap B_D,(F,y')\in\mathbb{G}, x'\ll y')\\
&\Leftrightarrow ( \exists F_1\in\mathcal{F}_{\tau})(F_1\sqsubseteq\lceil F\rceil,(F_1,x')\in\mathbb{G}\\
&\Leftrightarrow (F,x')\in\mathbb{G},
\end{align*}
 it follows  that $\mathbb{G}=\mathbb{G}_{f_{\mathbb{G}}}$.
\end{proof}
\subsection{Categorical equivalence}

The equivalence between categories demonstrates the essential identicalness between mathematical structures.
In this subsection, we establish the categorical equivalence between attribute continuous formal contexts and continuous domains. To make our result more transparent, we first recall some basic notions and results about category theory.
More notions of category theory can be  referred to~\cite{maclane71}.

Roughly speaking, a \emph{category} ${\bf C}$ consists of a collection of objects, a collection of morphism and a partial operator, named morphism, which obeys two laws. First, for each object~$A$ there exists an identity morphism~Id$_A$. Second, composition~$\circ$ is associative when defined.
The objects $A$ and $B$ in a category ${\bf C}$ is said to be  \emph{isomorphic} to each other if there
is a pair of morphisms $f : A \rightarrow B, g : B \rightarrow A $ such that $g \circ f = \text{Id}_A$ and
$f \circ  g =\text{Id}_B$.
A \emph{functor} from a category ${\bf C}$ to a category ${\bf D}$ is a map of objects to objects and morphisms to morphisms that
also preserves identities and composition.

Let $\mathcal{G}: {\bf C}\rightarrow {\bf D}$ be a functor. If, for all objects $A$ and $B$ of ${\bf C}$, the induced mapping $f\mapsto \mathcal{G}(f)$ between the set of morphisms from $A$ to $B$ and the set of morphisms from $\mathcal{G}(A)$ to $\mathcal{G}(B)$ is injective (respectively, surjective), then $\mathcal{G}$ is said to be \emph{faithful} (respectively, \emph{full}).

 We use the following well-known fact to prove the equivalence between two categories.

\begin{lemma}~\cite{maclane71}
Let ${\bf C}$ and ${\bf D}$ be two categories. Then ${\bf C}$ and ${\bf D}$ are categorically equivalent if and only if there exists a functor~$\mathcal{G}:{\bf C}\rightarrow{\bf D}$ such that $\mathcal{G}$ is full, faithful and essentially surjective on objects, that is for every object~$D$ of ${\bf D}$, there exists some object $C$ of ${\bf C}$ such that $\mathcal{G}(C)$ and $D$ are  isomorphic to each other.
\end{lemma}
Let $\mathbb{H}$ be an
$\mathcal{F}$-morphism from $\mathbb{P}_{\tau}=(P_o,P_a,\vDash, \mathcal{F}_{\tau})$ to $\mathbb{P}'_{\tau'}=(P'_o,P'_a,\vDash', \mathcal{F}_{\tau'}')$ and $\mathbb{H}'$ be an
$\mathcal{F}$-morphism from $\mathbb{P}_{\tau'}'=(P_o',P_a',\vDash', \mathcal{F}_{\tau'}')$ to $\mathbb{P}''_{\tau''}=(P''_o,P''_a,\vDash'', \mathcal{F}_{\tau''}'')$. Define $\mathbb{H}\circ \mathbb{H}'\subseteq \mathcal{F}_{\tau}\times P_a''$ by
\begin{equation}\label{e4.9}
(F,x'')\in\mathbb{H}\circ \mathbb{H}'\Leftrightarrow (\exists G\in\mathcal{F'}_{\tau'})(F\mathbb{H} G,(G,x'')\in\mathbb{H}'),
\end{equation}
and $\mathbb{I}_{P_a}\subseteq \mathcal{F}_{\tau}\times P_a$ by
\begin{equation}\label{e4.10}
(F,x)\in\mathbb{I}_{P_a} \Leftrightarrow x\in \lceil F\rceil.
\end{equation}
Routine checks verify that $\mathbb{H}\circ \mathbb{H}'$ is an $\mathcal{F}$-morphisms from $\mathbb{P}_{\tau}$ to $\mathbb{P}''_{\tau''}$ and
 $\mathbb{I}_{P_a}$ is an $\mathcal{F}$-morphism from $\mathbb{P}_{\tau}=(P_o,P_a,\vDash, \mathcal{F}_{\tau})$ to itself.

 Conditions (AR1) and (AR2) yield that the relation $\mathbb{I}_{P_a}$ is the identity morphism of $\mathbb{P}_{\tau}$. Using the same argument as checking the associative law of a traditional relation composition, it is easy to see that the composition~$\circ$  is also associative.

 Thus,
the attribute continuous formal contexts as objects with $\mathcal{F}$-morphisms as morphisms form a category ${\bf ACC}$.
We use ${\bf CD}$ to present the category of continuous domains with Scott-continuous functions.

In the rest of this section, we establish the categorical equivalence between ${\bf ACC}$ and ${\bf CD}$. To this end,  we first make a functor between them.

\begin{proposition}\label{p4.10}
$\mathcal{G}: {\bf ACC}\rightarrow {\bf CD}$ is a functor which maps every attribute continuous formal context $\mathbb{P}_{\tau}=(P_o,P_a,\vDash, \mathcal{F}_{\tau})$  to the continuous domain $(\mathfrak{B}(P_a,\mathcal{F}),\subseteq)$ and every $\mathcal{F}$-morphism $\mathbb{H}:\mathbb{P}_{\tau}\rightarrow \mathbb{P}'_{\tau'}$ to the Scott-continuous function $\phi_{\mathbb{H}}: \mathfrak{B}(P_a,\mathcal{F})\rightarrow \mathfrak{B}(P_a',\mathcal{F}')$, where $\phi_{\mathbb{H}}$ is defined by equation~(\ref{e4.3}).
\end{proposition}
\begin{proof}
By Theorems~\ref{t3.14} and~\ref{t4.5}, the corresponding objects map and morphism map are well-defined.
For any $Q\in\mathfrak{B}(P_a,\mathcal{F}_{\tau})$, we have
\begin{align*}
\mathcal{G}(\mathbb{I}_{P_a})(Q)&=\phi_{\mathbb{I}_{P_a}}(Q) \\
&=\set{x\in P_a}{(\exists F\in \mathcal{F}_{\tau})(F\sqsubseteq Q, x\in \lceil F\rceil})\\
&=\bigcup\set{\lceil F\rceil}{(\exists F\in \mathcal{F}_{\tau})(F\sqsubseteq Q)}\\
&=Q.
\end{align*}
This implies that $\mathcal{G}$ preserves the identity morphism.

 Let $\mathbb{H}$ be an
$\mathcal{F}$-morphism from $\mathbb{P}_{\tau}$ to $\mathbb{P}_{\tau'}'$ and $\mathbb{H}'$ an
$\mathcal{F}$-morphisms from $\mathbb{P}_{\tau'}'=(P_o',P_a',\vDash', \mathcal{F}_{\tau'}')$ to $\mathbb{P}''_{\tau''}=(P''_o,P''_a,\vDash'', \mathcal{F}_{\tau''}'')$.
For any $Q\in\mathfrak{B}(P_a,\mathcal{F})$ and $x''\in P_a''$, we have
\begin{align*}
x''\in \mathcal{G}(\mathbb{H}'\circ\mathbb{H})(Q)&\Leftrightarrow x''\in f_{\mathbb{H}'\circ\mathbb{H}}(Q) \\
&\Leftrightarrow (\exists F\in \mathcal{F}_{\tau})(F\subseteq Q,(F,x'')\in(\mathbb{H}'\circ\mathbb{H}))\\
&\Leftrightarrow (\exists F\in \mathcal{F}_{\tau},\exists F'\in \mathcal{F'}_{\tau'})(F\subseteq Q,F\mathbb{H} F',(F',x'')\in\mathbb{H}')\\
&\Leftrightarrow (\exists F'\in \mathcal{F}_{\tau'}')( F'\subseteq f_{\mathbb{H}}(Q),(F',x'')\in\mathbb{H}')\\
&\Leftrightarrow x''\in f_{\mathbb{H}'}(\mathcal{G}(\mathbb{H})(Q))\\
&\Leftrightarrow x''\in\mathcal{G}(\mathbb{H}')(\mathcal{G}(\mathbb{H})(Q)).
\end{align*}
This implies that $\mathcal{G}(\mathbb{H}'\circ\mathbb{H})=\mathcal{G}(\mathbb{H}')\circ\mathcal{G}(\mathbb{H})$, that is $\mathcal{G}$ preserves the composition.
\end{proof}
\begin{remark}

Similarly to Proposition~\ref{p4.10}, we can also prove that:
$\mathcal{H}: {\bf CD}\rightarrow {\bf ACC}$ is a functor which maps every continuous domain~$(D,\leq)$ to $Rep(D)=(B_D, D, \geq,\mathcal{F}_D),$
and for any Scott-continuous functions~$f : D \rightarrow D'$ to $\mathbb{G}_f,$
where $\mathbb{G}_f$ is defined  by equation~(\ref{e4.7}).
\end{remark}

\begin{theorem}\label{t4.12}
 The category~${\bf ACC}$ is equivalent to~${\bf CD}$.
\end{theorem}
\begin{proof}
According to Theorem~\ref{t3.14}, we see that the categories~${\bf ACC}$ and ${\bf CD}$ are essential surjective on objects. We now only need  to show that the functor~$\mathcal{G}$ defined in Proposition~\ref{p4.10} is full and faithful.

Let $\phi:\mathfrak{B}(P_a,\mathcal{F}_{\tau})\rightarrow \mathfrak{B}(P'_a,\mathcal{F}_{\tau'}')$ be a Scott-continuous function. From Theorem~\ref{t4.5}, we obtain an $\mathcal{F}$-morphism~$\mathbb{H}_{\phi}$ such that $\mathcal{G}(\mathbb{H}_{\phi})=\phi_{\mathbb{H}_{\phi}}=\phi$.
This implies that the functor $\mathcal{G}$ is full.

 Let $\mathbb{H}_1,\mathbb{H}_2: \mathbb{P}_{\tau} \rightarrow \mathbb{P}_{\tau'}'$ be $\mathcal{F}$-morphisms with $\phi_{\mathbb{H}_1}=\phi_{\mathbb{H}_2}$, where $\phi_{\mathbb{H}_1}$ and $\phi_{\mathbb{H}_2}$ are defined by equation~(\ref{e4.3}). For any $F\in \mathcal{F}_{\tau}$, we  have
\begin{align*}
(F,x')\in\mathbb{H}_1&\Leftrightarrow(\exists G\in\mathcal{F}_{\tau})(G\sqsubseteq F ,(G,x')\in\mathbb{H}_1)\\
&\Leftrightarrow x'\in\phi_{\mathbb{H}_1}(\lceil F\rceil)\\
&\Leftrightarrow x'\in\phi_{\mathbb{H}_2}(\lceil F\rceil)\\
&\Leftrightarrow(\exists G\in\mathcal{F}_{\tau})(G\sqsubseteq F ,(G,x')\in\mathbb{H}_2)\\
&\Leftrightarrow (F,x')\in\mathbb{H}_2.
\end{align*}
This implies that $\mathbb{H}_1=\mathbb{H}_2$, and hence the functor $\mathcal{G}$ is faithful.
\end{proof}

With $\mathcal{F}$-morphisms being morphisms, the subclasses of attribute continuous formal contexts introduced in Section 3.3 allows the building of special subcategories of ${\bf ACC}$. We write ${\bf DACC}$, ${\bf PACC}$, ${\bf TACC}$, ${\bf CACC}$ and ${\bf MACC}$ for the categories of dense attribute continuous formal contexts, of pointed attribute continuous formal contexts, of topped attribute continuous formal contexts, consistent attribute continuous formal contexts and multiplicative attribute continuous formal contexts, respectively. They are all full subcategories of ${\bf ACC}$. And we write ${\bf AD}$, ${\bf PCD}$, ${\bf TCD}$, ${\bf BCD}$ and ${\bf SCS}$ for the full subcategories of ${\bf CD}$ in which all objects are algebraic domains, pointed continuous domains, topped continuous domains, bounded complete domains and stably continuous semilattices, respectively.

 Similar to the proof of Theorem~\ref{t4.12}, it is not difficult to show that the categories~${\bf DACC}$, ${\bf PACC}$, ${\bf TACC}$, ${\bf CACC}$ and ${\bf MACC}$ are equivalent to~${\bf AD}$, ${\bf PCD}$, ${\bf TCD}$, ${\bf BCD}$ and ${\bf SCS}$, respectively.
\section{Conclusions and future work}

This paper has introduced  notions of attribute continuous context and continuous formal concept. It is shown that the set of continuous formal concepts of an attribute continuous formal context generate exactly the continuous domains, and  the category of attributes continuous formal contexts is equivalent to that of continuous domains. The results enrich the link between the two relatively independent areas: FCA and continuous domains, as outlined in~\cite{zhang06}.

As same as the classical FCA, there are  dual definitions  based on objects rather than attributes. Based on these dual definitions, we can also provide a representation for continuous domains.  Though the relationship between  continuous formal attribute concepts and  continuous formal object concepts is an interesting problem,
we do not investigate it in the paper.  Because it  has no bearing on with the subject at issue.

This paper has also studied the representations of a variety of subclasses of continuous domains, for example, algebraic domains, bounded complete domains and stably continuous semilattice. And the cases of continuous lattices and algebraic lattices can be obtained as a consequence.
It is worth mentioning that (1) L-domains and FS-domains can also be represented by attribute continuous formal contexts plus some additional requirements. But the proof of the representations  of these two subclasses are relatively complex. (2) For continuous lattices, there is a different representation in FCA from our mentioned above. We leave these two cases as a subject of a separate paper. It would be interesting to find appropriate conditions for other subclasses of continuous domains.

It also should be pointed out that there are  a number of open problems related to continuous semilattces. For example, a possible representation of continuous semilattices is still unknown. Proposition~\ref{p3.24} partially solves this problem in the case of the way-below relation being multiplicative. However, it is not difficult to illustrate that condition~(SS1) is a sufficient condition but not a necessary condition to force the generated domain  to be continuous semilattices, and condition~(SS2) is a sufficient condition but not a necessary condition to force the way-below relation to be multiplicative.
\bibliographystyle{plain}

\end{document}